\documentclass{article}

\usepackage{amsmath, amssymb}
\usepackage{graphicx}
\usepackage{wrapfig}
\usepackage{here}
\usepackage{amsthm}

\newtheorem{dfn}{Definition}
\newtheorem{lem}{Lemma}
\newtheorem{thm}{Theorem}
\newtheorem{prop}{Proposition}

\newtheorem{rmk}{Remark}
\newtheorem{cor}{Corollary}

\title{Cellular automata that generate symmetrical patterns give singular functions}
\author{Akane Kawaharada\footnote{E-mail: aka@kyokyo-u.ac.jp, Postal address: 1, Fujinomoricho, Fukakusa, Fushimi-ku, Kyoto-shi, Kyoto, 612-8522, Japan} \vspace{2mm}\\
Department of Mathematics, Kyoto University of Education}
\date{February 20, 2022}

\begin{document}

\maketitle

\begin{abstract}
In this paper, we mainly study linear one-dimensional and two-dimensional elementary cellular automata that generate symmetrical spatio-temporal patterns. 
For spatio-temporal patterns of cellular automata from the single site seed, we normalize the number of nonzero states of the patterns, take the limits, and give one-variable functions for the limit sets.
We can obtain a one-variable function for each limit set and show that the resulting functions are singular functions, which are non-constant, are continuous everywhere, and have a zero derivative almost everywhere.
We show that for Rule $90$, a one-dimensional elementary cellular automaton (CA), and a two-dimensional elementary CA, the resulting functions are Salem's singular functions.
We also discuss two nonlinear elementary CAs, Rule $22$, and Rule $126$.
Although their spatio-temporal patterns are different from that of Rule $90$, their resulting functions from the number of nonzero states equal the function of Rule $90$.
\end{abstract}

\hspace{2.5mm} 

\section{Introduction}
\label{sec:intro}

There are many studies about fractals generated by cellular automata, for example, Willson \cite{willson1984}, Culik and Dube \cite{culik1989}, Takahashi \cite{takahashi1992}, and Haeseler et al. \cite{haeseler1993}.
A cellular automaton (CA) is a discrete dynamical system, in some cases whose spatio-temporal pattern from a single site seed holds self-similarity and whose limit set is a fractal.
In general, when we characterize the fractal, we calculate its fractal dimension, given by a specific numerical value.
In this study, however, we assign a real one-variable function for the fractal to capture details of the fractal structure.

In this paper, we study linear elementary CAs that generate symmetrical spatio-temporal patterns. 
For linear automata, their spatio-temporal patterns from single site seeds hold self-similarity or partial self-similarity, and we can calculate the number of nonzero states of the patterns using the structure.
We obtained results about the numbers for some elementary CAs in our previous paper \cite{kawanami2020}.
From those results, we normalize the numbers of nonzero states, take the limits, and provide one-variable functions for limit sets.
We show that if a CA is linear, it holds a counting equation in Lemma~\ref{lem:fin}, which counts the number of nonzero states of the spatio-temporal pattern, and, using the equation, we obtain a one-variable function for each limit set of the automaton.
This means one-variable functions characterize fractals by projecting fractals onto one-variable functions.

Next, we discuss the properties of the obtained one-variable functions and give sufficient conditions for the singularity of a function in Theorem~\ref{thm:sing}.
We also show that the obtained functions are singular functions, which are monotonically increasing (or decreasing), are continuous everywhere, and have a zero derivative almost everywhere.
For the one-dimensional elementary CA Rule $90$, the resulting function equals Salem's singular function $L_{1/3}$, a self-affine function \cite{salem1943, derham1957, ulam1934, yhk1997}, and for a two-dimensional elementary CA, the resulting function equals Salem's singular function $L_{1/5}$ (numerical results were obtained in \cite{kawa2014a, kawanami2014} where we showed that the difference forms of the equations match Salem's in \cite{kawanami2020}).
For the one-dimensional CA Rule $150$, we previously demonstrated that the resulting function is a singular function that strictly increases, is continuous, and is differentiable almost everywhere \cite{kawa2021}.
In addition, we discuss two nonlinear elementary CAs, Rule $22$ and Rule $126$.
Their spatio-temporal patterns are similar to that of Rule $90$ (see Figures~\ref{fig:90p}, \ref{fig:22p}, and \ref{fig:126p}).
We show that their resulting functions from the number of nonzero states equal the function of Rule $90$.

The remainder of this paper is organized as follows. 
Section~\ref{sec:pre} describes the preliminaries concerning CAs and the previous results about the number of nonzero states of spatial and spatio-temporal patterns of CAs. 
For linear one-dimensional and two-dimensional elementary CAs, Section~\ref{sec:main} reports our main results about real one-variable functions given using the self-similarities of the spatio-temporal patterns.
We also provide sufficient conditions for the singularity of a function and
show that the obtained functions are singular.
Further, we discuss two nonlinear one-dimensional elementary CAs, Rule $22$ and Rule $126$, whose normalized functions equal that of Rule $90$.
Finally, Section~\ref{sec:cr} discusses the findings of this paper and highlights possible avenues for future studies.

%
%

\section{Preliminaries}
\label{sec:pre}

\subsection{Definitions and notations}
\label{subsec:def}

We provide some definitions and notations about CA.
Let $\{0, 1\}$ be a state set and $\{0,1\}^{{\mathbb Z}^d}$ be a $d$-dimensional configuration space for $d \in {\mathbb Z}_{>0}$.
We define a configuration $u_o \in \{0, 1\}^{{\mathbb Z}^d}$ as
\begin{align}
\label{eq:sss}
({u_o})_{{\textit{\textbf i}}} = \left\{
\begin{array}{ll}
1 & \mbox{if ${\textit{\textbf i}} = {\bf 0} \ ( = (0, 0, \ldots, 0))$},\\
0 & \mbox{if ${\textit{\textbf i}} \in {{\mathbb Z}^d} \backslash \{ \bf 0 \}$}.
\end{array}
\right. 
\end{align}
We call $u_o$ the single site seed.
Let $(\{0, 1\}^{{\mathbb Z}^d}, S)$ be a discrete dynamical system for a transformation $S$ on $\{0,1\}^{{\mathbb Z}^d}$.
The $n$-th iteration of $S$ is denoted by $S^n$.

We define one-dimensional and two-dimensional elementary CAs as follows.

\begin{dfn}
\label{dfn:1deca}
\begin{enumerate}
\item[$(i)$] A one-dimensional elementary cellular automaton ($1$dECA) $(\{0, 1\}^{\mathbb Z}, S)$ is given by 
\begin{align}
(S u)_i &= s(u_{i-1}, u_i, u_{i+1}) 
\end{align}
for $i \in {\mathbb Z}$ and $u \in \{0, 1\}^{\mathbb Z}$,
where $s : \{0,1\}^{3} \to \{0,1\}$ is a local rule of $S$.
\item[$(ii)$] A $1$dECA $(\{0, 1\}^{\mathbb Z}, S)$ is a symmetrical pattern generation $1$dECA (SPG$1$dECA) if a local rule $s$ satisfies $s(u_{i-1}, u_i, u_{i+1}) = s(u_{i+1}, u_i, u_{i-1})$, $s(0, 0, 0)=0$, and $s(0, 0, 1)=1$.
\item[$(iii)$] A $1$dECA $(\{0,1\}^{\mathbb Z}, S)$ is linear if a local rule satisfies
\begin{align}
(S u)_i = c_0 u_{i-1} + c_1 u_i + c_2 u_{i+1} \quad \mbox{(mod $2$)},
\end{align}
where $c_0$, $c_1$, $c_2 \in \{0, 1\}$.
\end{enumerate}
\end{dfn}

\begin{rmk}
\label{rmk:1dlin}
There exist $256$ $1$dECAs and $16$ of them are SPG$1$dECAs. 
Only two SPG$1$dECAs, Rule $90$ and Rule $150$, are linear SPG$1$dECAs.
\end{rmk}

Table~\ref{tab:1dECA} shows the local rules of SPG$1$dECAs Rule $22$, Rule $90$, Rule $126$, and Rule $150$.

\begin{table}[H]
\caption{Local rules of SPG$1$dECAs}
\label{tab:1dECA}
\begin{center}
\begin{tabular}{c | c c c c c c}
\hline 
$u_{i-1} u_i u_{i+1}$ & $111$ & $110$ & $101$ & $100$ & $010$ & $000$ \\ 
\hline
$(S_{22} u)_i$ & $0$ & $0$ & $0$ & $1$ & $1$ & $0$ \\
$(S_{90} u)_i$ & $0$ & $1$ & $0$ & $1$ & $0$ & $0$ \\
$(S_{126} u)_i$ & $0$ & $1$ & $1$ & $1$ & $1$ & $0$ \\
$(S_{150} u)_i$ & $1$ & $0$ & $0$ & $1$ & $1$ & $0$ \\
\hline
\end{tabular}
\end{center}
\end{table}

\begin{dfn}
\label{dfn:2deca}
\begin{enumerate}
\item[$(i)$] A two-dimensional elementary cellular automaton ($2$dECA) $(\{0,1\}^{{\mathbb Z}^2}, T)$ is given by 
\begin{align}
(T u)_{i, j} = 
t 
\begin{pmatrix}
u_{i, j+1} \\ u_{i-1, j} \quad u_{i, j} \quad u_{i+1, j} \\ u_{i, j-1}
\end{pmatrix}
= t
\begin{pmatrix}
U \\ L C R \\ D
\end{pmatrix}
\end{align}
for $(i, j) \in {\mathbb Z}^{2}$ and $u \in \{0,1\}^{{\mathbb Z}^{2}}$,  
where $t : \{0,1\}^{5} \to \{0,1\}$ is a local rule depending on the five states of the von Neumann neighborhood. 
\item[$(ii)$] A $2$dECA $(\{0, 1\}^{{\mathbb Z}^2}, T)$ is a symmetrical pattern generation $2$dECA (SPG$2$dECA) if a local rule $t$ satisfies 
\begin{align}
\label{eq:sym1}
&t \! \! 
\begin{pmatrix}
U \\  \! L C R \! \\ D
\end{pmatrix}
= t \! \! 
\begin{pmatrix}
D \\ \! L C R \! \\ U
\end{pmatrix}
= t \! \! 
\begin{pmatrix}
U \\ \! R C L \! \\ D
\end{pmatrix}
, \\
\label{eq:sym2}
&t \! \! 
\begin{pmatrix}
U \\ L C R \\ D
\end{pmatrix}
= t \! \! 
\begin{pmatrix}
L \\ \! D C U \! \\ R
\end{pmatrix}
= t \! \! 
\begin{pmatrix}
D \\ \! R C L \! \\ U
\end{pmatrix}
= t \! \! 
\begin{pmatrix}
R \\ \! U C D \! \\ L
\end{pmatrix},\\
&t \! \! 
\begin{pmatrix}
0 \\ 0 0 0 \\ 0
\end{pmatrix}
= 0, \
t \! \! 
\begin{pmatrix}
0 \\ \! 0 0 0 \! \\ 1
\end{pmatrix}
= 1.
\end{align}
\item[$(iii)$] A $2$dECA $(\{0,1\}^{{\mathbb Z}^2}, T)$ is linear if a local rule satisfies
\begin{align}
(T u)_{i, j} = c_0 u_{i, j} + c_1 u_{i+1, j} + c_2 u_{i, j-1} + c_3 u_{i-1, j} + c_4 u_{i, j+1} \quad \mbox{(mod $2$)},
\end{align}
where $c_k \in \{0, 1\}$, $k=0, 1, \ldots, 4$.
\end{enumerate}
\end{dfn}

\begin{rmk}
\label{rmk:2dlin}
There exist $2^{32}$ $2$dECAs and $2^{10}$ of them are SPG$2$dECAs.
Only two SPG$2$dECAs, $T_0$ and $T_{528}$, are linear SPG$2$dECAs.
\end{rmk}

Table~\ref{tab:2dECA} shows the local rules of SPG$2$dECAs $T_{0}$ and $T_{528}$.

\begin{table}[H]
\caption{Local rules of SPG$2$dECAs}
\label{tab:2dECA}
\begin{center}
\scalebox{0.9}{
\begin{tabular}{c | c c c c c c c c c c c c} \hline \vspace{-1mm}
$U$ & $1$ & $0$ & $1$ & $0$ & $0$ & $0$ & $1$ & $0$ & $1$ & $0$ & $0$ & $0$\\ \vspace{-1mm}
$LCR$ & $111$ & $111$ & $010$ & $011$ & $010$ & $010$  &  $101$  &  $101$  &  $000$  &  $001$  &  $000$  & $000$\\
$D$ & $1$ & $1$ & $1$ & $1$ & $1$ & $0$ & $1$ & $1$ & $1$ & $1$ & $1$ & $0$\\ \hline
$(T_{0} u)_{i,j}$
& $0$ & $1$ & $0$ & $0$ & $1$ & $0$ & $0$ & $1$ & $0$ & $0$ & $1$ & $0$ \\
$(T_{528} u)_{i,j}$
& $1$ & $0$ & $1$ & $1$ & $0$ & $1$ & $0$ & $1$ & $0$ & $0$ & $1$ & $0$ \\ 
\hline
\end{tabular}
}
\end{center}
\end{table}

Next, we introduce a singular function related to the spatio-temporal patterns of CAs.

\begin{dfn}[\cite{derham1957, yhk1997}]
\label{def:lb}
Let $\alpha$ be a parameter such that $0 < \alpha < 1$ and $\alpha \neq 1/2$.
The singular function $L_{\alpha}:[0,1] \to [0,1]$ is defined as follows: 
\begin{align}
\label{eq:lb}
L_{\alpha} (x):=
\left\{
\begin{array}{ll}
\alpha L_{\alpha}(2 x) & \ (0 \leq x < 1/2),\\
(1-\alpha) L_{\alpha}(2 x-1) + \alpha & \ (1/2 \leq x \leq 1).
\end{array}
\right. 
\end{align}
\end{dfn}
The functional equation has a unique continuous solution on the unit interval $[0,1]$.
The resulting function $L_{\alpha}$ is strictly increasing, is continuous, and has the derivative zero almost everywhere.
A difference form of the function $L_{\alpha}$ is given by 
\begin{align}
\label{eq:diflb}
L_{\alpha} \left( \frac{2i+1}{2^{k+1}} \right) 
&= (1-\alpha) L_{\alpha} \left( \frac{i}{2^{k}} \right) 
+ \alpha L_{\alpha} \left( \frac{i+1}{2^{k}} \right)
\end{align}
for $0 \leq i \leq 2^k-1$, $k \in {\mathbb Z}_{>0}$ \cite{yhk1997}.
The end points are given by $L_{\alpha}(0)=0$, $L_{\alpha}(1)=1$.

\subsection{Previous results about the number of nonzero states of linear SPG$1$dECAs and a linear SPG$2$dECA}
\label{subsec:pre}

We introduce some previous results about the number of nonzero states in spatial and spatio-temporal patterns of linear SPG$1$dECAs and a linear SPG$2$dECA.

For a CA $(\{0, 1\}^{{\mathbb Z}^d}, T)$, a subset of a $(d+1)$-dimensional Euclidean space $V_T(n)$ is given by
\begin{align}
V_T(n) = \{ ({\bf i}, m) \in {\mathbb Z}^{d+1} \mid (T^m u_o)_{\bf i} > 0, 0 \leq m \leq n \},
\end{align}
which consists of nonzero states from time step $0$ to $n$.
Let $V_T(n)/n$ be a contracted set of $V_T(n)$ with a contraction rate of $1/n$. 
A limit set of a CA is then defined by $\lim_{n \to \infty} (V_T(n)/n)$ if it exists. 
For the limit sets of linear CAs, the following two theorems have been reported.

\begin{thm}[\cite{takahashi1992}]
\label{thm:tkhs01}
Let $p$ be a prime number and $m \in {\mathbb Z}_{>0}$.
For a $p^m$-state linear CA, if $p^{m-1}$ divides time step $n$, then $(T^{pn} u_o)_{pi} = (T^n u_o)_i$.
If $p^m$ divides $n$ and at least one of the elements of $i$ is indivisible by $p$,
then $(T^n u_o)_i$ equals $0$.
\end{thm}

\begin{thm}[\cite{takahashi1992}]
\label{thm:tkhs02}
Let $p$ be a prime number and $m \in {\mathbb Z}_{>0}$.
For a $p^m$-state linear CA, its limit set $\lim_{k \to \infty} (V_T(p^k-1)/p^k)$ exists.
\end{thm}

From Theorems~\ref{thm:tkhs01} and \ref{thm:tkhs02}, we obtain the following results about linear CAs.
 
For a CA $(\{0, 1\}^{{\mathbb Z}^d}, T)$, let $num_T(n)$ be the number of nonzero states in a spatial pattern $T^n u_o$ for time step $n$, and let $cum_T(n)$ be the cumulative sum of the number of nonzero states in a spatial pattern $T^m u_o$ from time step $m=0$ to $n$. Thus,
\begin{align}
num_T(n) = \sum_{{\textit{\textbf i}} \in {{\mathbb Z}^d}} (T^n u_o)_{\textit{\textbf i}}, \
cum_T(n) = \sum_{m = 0}^n \sum_{{\textit{\textbf i}} \in {{\mathbb Z}^d}} (T^m u_o)_{\textit{\textbf i}},
\end{align}
where $cum_T(-1) = num_T(-1) = 0$.

For SPG$1$dECA Rule $90$ and SPG$2$dECA $T_0$, we obtained the following results.

\begin{prop}[\cite{kawanami2020}]
\label{prop:jmp}
\begin{enumerate}
\item[$(i)$] For SPG$1$dECA Rule $90$, let $x = j/2^{k+1}$ for $0 \leq j \leq 2^{k+1}$, $k \in {\mathbb Z}_{>0}$, and $h_{S90}(x) = cum_{S90}(j-1)/cum_{S90}(2^{k+1}-1)$. Thus,
\begin{align}
h_{S90} \left( \frac{2i+1}{2^{k+1}} \right) 
&= \left( 1- \frac{1}{3} \right) h_{S90} \left( \frac{2i}{2^{k+1}} \right) + \frac{1}{3} h_{S90} \left( \frac{2i+2}{2^{k+1}} \right) 
\end{align}
for $0 \leq i \leq 2^k-1$.
The boundary conditions are given by $h_{S90}(0)=0$ and $h_{S90}(1)=1$.
\item[$(ii)$] For SPG$2$dECA $T_0$, let $x = j/2^{k+1}$ for $0 \leq j \leq 2^{k+1}$, $k \in {\mathbb Z}_{>0}$, and $h_{T0}(x) = cum_{T0}(j-1)/cum_{T0}(2^{k+1}-1)$. Thus,
\begin{align}
h_{T0} \left( \frac{2i+1}{2^{k+1}} \right) 
&= \left( 1- \frac{1}{5} \right) h_{T0} \left( \frac{2i}{2^{k+1}} \right) + \frac{1}{5} h_{T0} \left( \frac{2i+2}{2^{k+1}} \right) 
\end{align}
for $0 \leq i \leq 2^k-1$.
The boundary conditions are given by $h_{T0}(0)=0$ and $h_{T0}(1)=1$.
\end{enumerate}
\end{prop}

Therefore, $h_{S90}$ and $h_{T0}$ equal the difference forms of the singular functions $L_{1/3} (x)$ and $L_{1/5} (x)$, respectively.

For SPG$1$dECA Rule $150$, we obtained the following result.

\begin{prop}[\cite{kawa2021}]
\label{prop:fx150}
For $x = \sum_{i=1}^{\infty} ( x_i / 2^i ) \in [0,1]$, the function $f_{S150} : [0,1] \to [0,1]$ is given by
\begin{align}
\label{eq:def}
f_{S150} (x) &= \lim_{k \to \infty} \frac{cum_{S150}\left((\sum_{i=1}^k x_i 2^{k-i})-1 \right)}{cum_{S150}(2^k-1)}\\
&= \sum_{i=1}^{\infty} x_i \alpha^i \prod_{s=0}^{i-1} \left( \frac{(-1)^{s+1} + 2^{s+2}}{3} \right)^{p_{i,s}},
\end{align}
where $\alpha = (\sqrt{5}-1)/4$ and $p_{i, s}$ is the number of clusters consisting of $s$ continuous $1$s in the binary number $0.x_1 x_2 \cdots x_{i-1}$.
\end{prop}

\section{Main results}
\label{sec:main}

We discuss the main results concerning SPG$1$dECAs and SPG$2$dECAs.
Among the SPG$1$dECAs, Rule $90$ and Rule $150$ are linear and hold the equation in Lemma~\ref{lem:fin}.
Among SPG$2$dECAs, only two CAs, $T_0$ and $T_{528}$, are both linear and hold with the equation in Lemma~\ref{lem:fin}.
In Section~\ref{subsec:selfsim}, we calculate the number of nonzero states of the spatial and spatio-temporal patterns of Rule $90$, $T_0$, and $T_{528}$.
We normalized the dynamics of the number of nonzero states and obtained functions for them. 
(For Rule $150$, we previously obtained the results in \cite{kawa2021}.)
In Section~\ref{subsec:rx}, we provide a sufficient condition of singularity for a function and show that the resulting functions for the four CAs are singular functions, which are strictly increasing, continuous, and differentiable with the derivative zero almost everywhere.
We also show that $f_{S90}$ and $f_{T0}$ are Salem's singular function $L_{1/\alpha}$, and the box-counting dimension of their limit sets are given by $-\log \alpha/\log 2$. 
In Section~\ref{subsec:NONlin}, we discuss the nonlinear SPG$1d$ECAs, Rule $22$ and Rule $126$. 
We discovered that the normalized functions for Rule $22$ and Rule $126$ equal the function for Rule $90$.


\subsection{Singular functions generated by linear SPG ECAs}
\label{subsec:selfsim}

For spatio-temporal patterns of an SPG$1$dECA Rule $90$ and SPG$2$dECAs, $T_0$ and $T_{528}$, we calculate the number of nonzero states, $cum_T$ and $num_T$, and provide normalized functions.

Let $LS_1$ be the set of linear SPG$1$dECAs, and let $LS_2$ be the set of linear SPG$2$dECAs.
By Theorem~\ref{thm:tkhs01}, for a CA in $LS_1$ and $LS_2$, we can count the number of (partially) self-similar sets in each spatio-temporal pattern for time step $n$.
Then, we obtain the following lemma.

\begin{lem}
\label{lem:fin}
Let $n = \sum_{i=0}^{k-1} x_{k-i} \, 2^i \geq 0$, where $x_0 = 0$. 
If a CA $T \in LS_1 \cup LS_2$, then
\begin{align}
cum_T (n-1) 
&= \sum_{i=1}^k x_i \ num_T \left( \sum_{j=0}^{i-1} x_j 2^{k-j} \right) cum_T (2^{k-i}-1).
\end{align}
\end{lem}

By Theorem~\ref{thm:tkhs02} for CAs in $LS_1 \cup LS_2$, the following function $f_T: [0, 1] \to [0, 1]$ exists.

\begin{dfn}
\label{dfn:lim}
For a CA $(\{0, 1\}^{{\mathbb Z}^d}, T)$, a function $f_T: [0, 1] \to [0, 1]$
is given by
\begin{align}
\label{eq:org}
f_T (x) &:= \lim_{k \to \infty} \frac{cum_T\left((\sum_{i=1}^k x_i 2^{k-i})-1 \right)}{cum_T(2^k-1)}
\end{align}
for $x = \sum_{i=1}^{\infty} (x_i/2^i) \in [0,1]$.
\end{dfn}

Next, we consider functions $f_{S90}$, $f_{T0}$, and $f_{T528}$.

\subsubsection{Function $f_{S90}$ generated by Rule $90$}

Figure~\ref{fig:90p} shows the spatio-temporal pattern of Rule $90$ from the single site seed $u_o$, and Figure~\ref{fig:90g} shows the graph of the cumulative number of nonzero states in the spatio-temporal pattern of Rule $90$.
The values $cum_{S90}(2^k-1)$ and $num_{S90}(n)$ were already obtained.

\begin{lem}[\cite{kawanami2020}]
\label{lem:90num}
For time step $n = \sum_{i=0}^{k-1} x_{k-i} 2^i$, we have $cum_{S90}(2^k-1) = 3^k$ and $num_{S90}(n) = 2^{\sum_{i=1}^k x_i}$.
\end{lem}

\begin{figure}[htbp]
\begin{minipage}{0.45\linewidth}
\centering
\includegraphics[width=\linewidth]{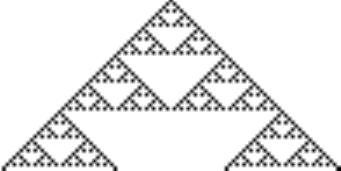}
\caption{Spatio-temporal pattern of Rule $90$}
\label{fig:90p}
\end{minipage}
\begin{minipage}{0.1\linewidth}
\quad
\end{minipage}
\begin{minipage}{0.45\linewidth}
\centering
\includegraphics[width=\linewidth]{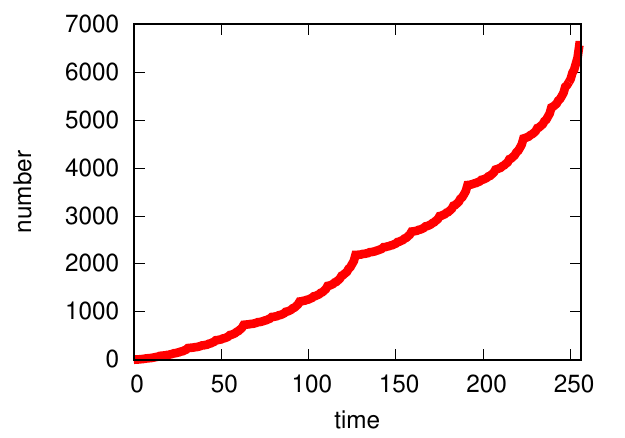}
\caption{Dynamics of the cumulative number of nonzero states of Rule $90$}
\label{fig:90g}
\end{minipage}
\end{figure}

\begin{thm}
\label{thm:fxS90}
For $x = \sum_{i=1}^{\infty} ( x_i / 2^i ) \in [0,1]$, the function $f_{S90} : [0,1] \to [0,1]$ is given by
\begin{align}
\label{eq:defS90}
f_{S90}(x) = \sum_{i=1}^{\infty} x_i \, 2^{\sum_{j=0}^{i-1} x_j} \, 3^{-i}.
\end{align}
\end{thm} 
 
\begin{proof}
By Lemmas~\ref{lem:fin} and \ref{lem:90num}, we have
\begin{align}
f_{S90} (x) 
&= \lim_{k \to \infty} \frac{cum_{S90}\left((\sum_{i=1}^k x_i 2^{k-i})-1 \right)}{cum_{S90}(2^k-1)} \\
&= \lim_{k \to \infty} \frac{\sum_{i=1}^k x_i \ num_{S90} \left( \sum_{j=0}^{i-1} x_j 2^{k-j} \right) cum_{S90} (2^{k-i}-1)}{cum_{S90} (2^k-1)} \\
&= \lim_{k \to \infty} \sum_{i=1}^k x_i \ 2^{\sum_{j=0}^{i-1} x_j} 3^{-i}. \label{eq:convS90}
\end{align}

From Equation~\eqref{eq:convS90}, 
we have $x_i \ 2^{\sum_{j=0}^{i-1} x_j} 3^{-i} \leq (2/3)^i$.\\
Because $\lim_{i \to \infty} \left| (2/3)^{i+1}/(2/3)^i \right| < 1$, the infinite series $\sum_{i=1}^{\infty} (2/3)^i$ absolutely converges.
Thus, $\sum_{i=1}^{\infty} x_i \ 2^{\sum_{j=0}^{i-1} x_j} 3^{-i}$ also absolutely converges.

We easily obtain $f_{S90}(0)=0$ and $f_{S90}(1)=1$.
Therefore, Equation~\eqref{eq:defS90} is obtained.
\end{proof}

\begin{rmk}
When $x$ is a dyadic rational, $m/2^i$, we have two possible binary expansions.
We will verify that the definition of $f_{S90}$ is consistent for the values
with two binary expansions.
Let $x = \sum_{i=1}^k (x_i / 2^i) + 1 / 2^{k+1}$ and
$y = \sum_{i=1}^k ( x_i / 2^i) + \sum_{i=k+2}^{\infty} (1/2^i)$ for $x_i \in \{0,1\}$ and $k \in {\mathbb Z}_{>0}$. Hence, $x=y$.
We have
\begin{align}
f_{S90}(y)-f_{S90}(x) 
&= \left( \sum_{i=k+2}^{\infty} 2^{i-k-2+\sum_{j=1}^k x_j} \, 3^{-i} \right) - 2^{\sum_{j=1}^k x_j} \, 3^{-k-1} \\
&= 2^{\sum_{j=1}^k x_j} \, 3^{-i} \left( \sum_{i=1}^{\infty} 2^{i-1} \, 3^{-i} -1 \right) = 0.
\end{align}
\end{rmk}

\subsubsection{Function $f_{T0}$ generated by $T_0$}

Figure~\ref{fig:0p} shows the spatio-temporal pattern of an SPG$2$dECA $T_0$ from the single site seed $u_o$, and Figure~\ref{fig:0g} shows the graph of the cumulative number of nonzero states in the spatio-temporal pattern of $T_0$.
The values $cum_{T0}(2^k-1)$ and $num_{T0}(n)$ were already obtained.

\begin{lem}[\cite{kawanami2020}]
\label{lem:T0num}
For time step $n = \sum_{i=0}^{k-1} x_{k-i} 2^i$, we have $cum_{T0}(2^k-1) = 5^k$ and $num_{T0}(n) = 4^{\sum_{i=1}^k x_i}$.
\end{lem}

\begin{figure}[htbp]
\centering
\includegraphics[width=\linewidth]{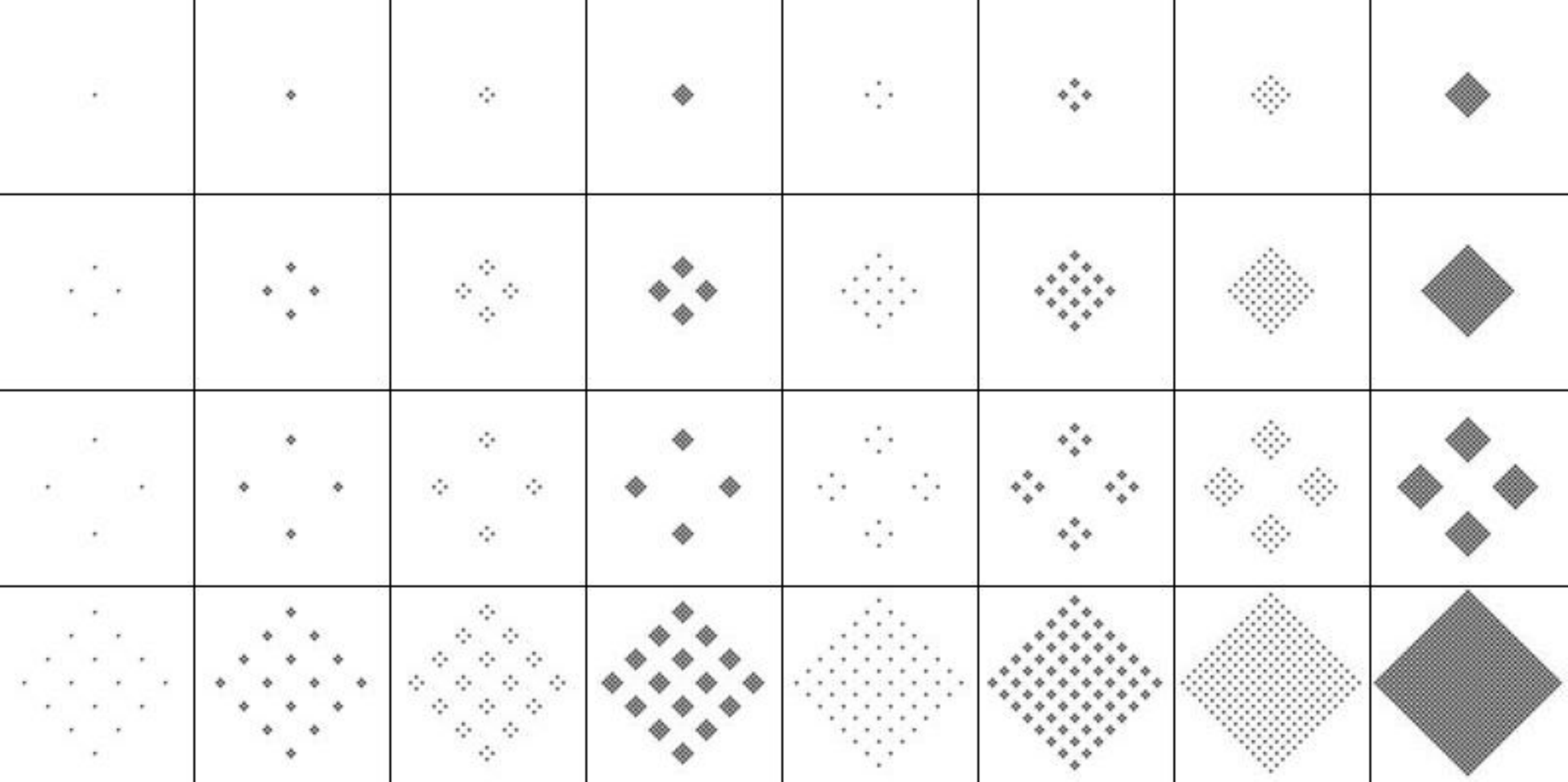}
\caption{Spatio-temporal pattern of $T_0$}
\label{fig:0p}
\end{figure}

\begin{thm}
\label{thm:fxT0}
For $x = \sum_{i=1}^{\infty} ( x_i / 2^i ) \in [0, 1]$, the function $f_{T0} : [0, 1] \to [0, 1]$ is given by
\begin{align}
\label{eq:defT0}
f_{T0}(x) = \sum_{i=1}^{\infty} x_i \, 4^{\sum_{j=0}^{i-1} x_j} \, 5^{-i}.
\end{align}
\end{thm} 
 
\begin{proof}
By Lemmas~\ref{lem:fin} and \ref{lem:T0num}, we have
\begin{align}
f_{T0} (x) 
&= \lim_{k \to \infty} \frac{cum_{T0}\left((\sum_{i=1}^k x_i 2^{k-i})-1 \right)}{cum_{T0}(2^k-1)} \\
&= \lim_{k \to \infty} \frac{\sum_{i=1}^k x_i \ num_{T0} \left( \sum_{j=0}^{i-1} x_j 2^{k-j} \right) cum_{T0} (2^{k-i}-1)}{cum_{T0} (2^k-1)} \\
&= \lim_{k \to \infty} \frac{1}{5^k} \sum_{i=1}^k x_i \ num_{T0} \left( \sum_{j=0}^{i-1} x_j 2^{k-j} \right) 5^{k-i}\\
&= \lim_{k \to \infty} \sum_{i=1}^k x_i \ num_{T0} \left( \sum_{j=0}^{i-1} x_j 2^{k-j} \right) 5^{-i}\\
&= \lim_{k \to \infty} \sum_{i=1}^k x_i \ 4^{\sum_{j=0}^{i-1} x_j} 5^{-i}. \label{eq:convT0}
\end{align}

From Equation~\eqref{eq:convT0}, we have $x_i \ 4^{\sum_{j=0}^{i-1} x_j} 5^{-i} \leq (4/5)^i$. \\
Because $\lim_{i \to \infty} \left| (4/5)^{i+1}/(4/5)^i \right| = 4/5 < 1$, the infinite series $\sum_{i=1}^{\infty} (4/5)^i$ absolutely converges.
Thus, $\sum_{i=1}^{\infty} x_i \ 4^{\sum_{j=0}^{i-1} x_j} 5^{-i}$ also absolutely converges.

We easily obtain $f_{T0}(0)=0$ and $f_{T0}(1)=1$.
Therefore, Equation~\eqref{eq:defT0} is obtained.
\end{proof}

\begin{rmk}
For $x = \sum_{i=1}^k (x_i / 2^i) + 1 / 2^{k+1}$ and
$y = \sum_{i=1}^k ( x_i / 2^i) + \sum_{i=k+2}^{\infty} (1/2^i)$, for $x_i \in \{0,1\}$ and $k \in {\mathbb Z}_{>0}$, we verify that $f_{T0}(x) = f_{T0}(y)$ because $x=y$. We have
\begin{align}
f_{T0}(y)-f_{T0}(x) 
&= \left( \sum_{i=k+2}^{\infty} 4^{i-k-2+\sum_{j=1}^k x_j} \, 5^{-i} \right) - 4^{\sum_{j=1}^k x_j} \, 5^{-k-1} \\
&= 4^{\sum_{j=1}^k x_j} \, 5^{-i} \left( \sum_{i=1}^{\infty} 4^{i-1} \, 5^{-i} -1 \right) = 0.
\end{align}
\end{rmk}

\begin{figure}[htbp]
\begin{minipage}{0.47\linewidth}
\centering
\includegraphics[width=\linewidth]{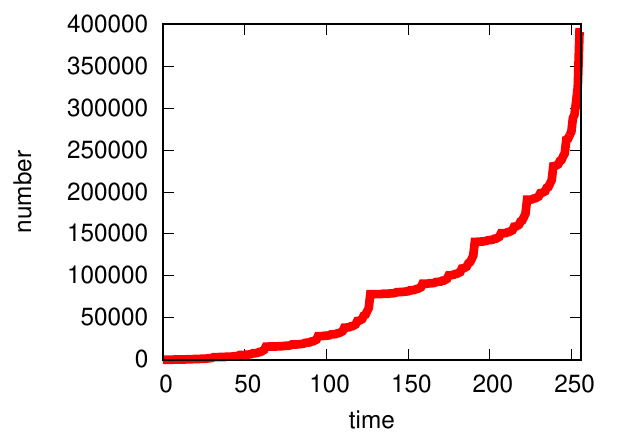}
\caption{Dynamics of the cumulative number of nonzero states of $T_0$}
\label{fig:0g}
\end{minipage}
\begin{minipage}{0.06\linewidth}
\quad
\end{minipage}
\begin{minipage}{0.47\linewidth}
\centering
\includegraphics[width=\linewidth]{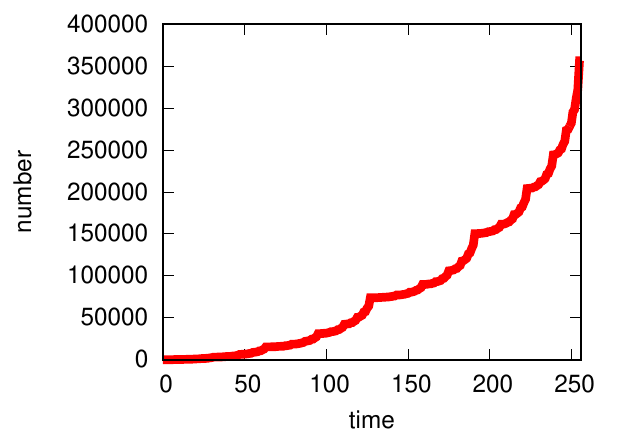}
\caption{Dynamics of the cumulative number of nonzero states of $T_{528}$}
\label{fig:528g}
\end{minipage}
\end{figure}

\subsubsection{Function $f_{T528}$ generated by $T_{528}$}

We study the spatio-temporal pattern of $(\{0, 1\}^{{\mathbb Z}^2}, T_{528})$ from the initial configuration $u_o$ (Figure~\ref{fig:528p}).
Figure~\ref{fig:528g} shows the cumulative number of nonzero states of $T_{528}$. 
First, we obtain $cum_{T528}(2^k-1)$ and $num_{T528}(n)$.

\begin{lem}
\label{lem:T528num}
For time step $n = \sum_{i=0}^{k-1} x_{k-i} 2^i$, we have
\begin{align}
& cum_{T528}(2^k-1) = 2^{k-1} \left((1 + \sqrt{2})^{k+1} + (1 - \sqrt{2})^{k+1} \right),\\
& num_{T528}(n) = \prod_{s=0}^{l+1} \left( \frac{17 + 7 \sqrt{17}}{34} \left( \frac{3 + \sqrt{17}}{2} \right)^s + \frac{17 - 7 \sqrt{17}}{34} \left( \frac{3 - \sqrt{17}}{2} \right)^s \right)^{p_s},
\end{align}
where $p_s$ is the number of clusters of $s$ consecutive $1$-states in the binary number of $n$, $x_1 x_2 \ldots x_k$.
\end{lem}

\begin{figure}[htbp]
\centering
\includegraphics[width=\linewidth]{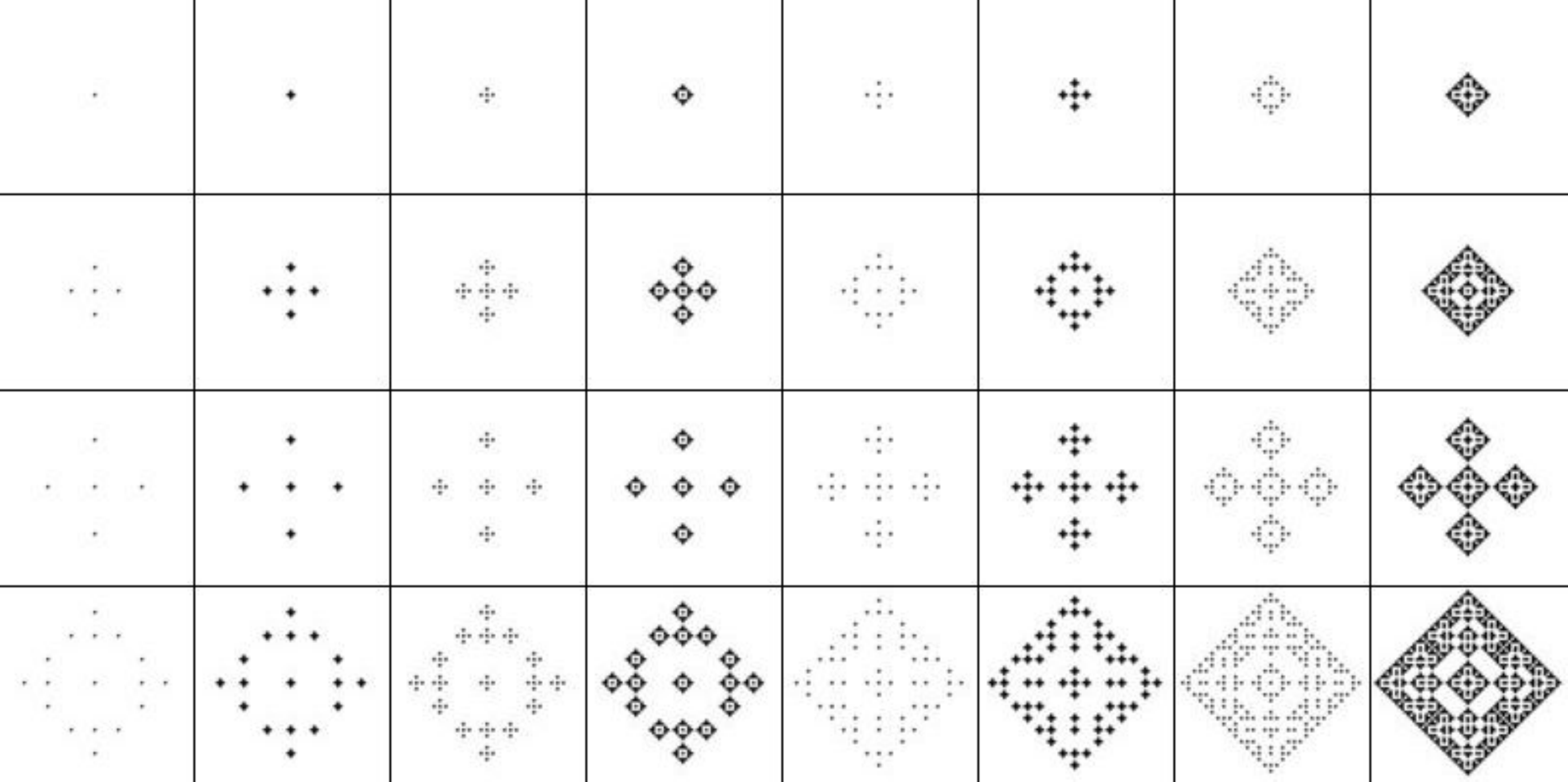}
\caption{Spatio-temporal pattern of $T_{528}$}
\label{fig:528p}
\end{figure}

\begin{proof}
For some $k \in {\mathbb Z}_{\geq 0}$ we study a set of nonzero states of the spatio-temporal pattern $\{T_{528}^n u_o\}_{n=0}^{2^k-1}$, $V_{T528}(2^k-1)$.
We provide three types of partially self-similar sets, $A_k$, $B_k$, and $C_k$, based on the four-sided pyramid $V_{T528}(2^k-1)$ (see Figure~\ref{fig:ABC}).
Let $A_k$ be $V_{T528}(2^k-1)$ itself.
We remove a quadrangular prism whose size is $1 \times 1 \times (2^k-1)$ from $A_k$ and crop it vertically to quarter of its size through the top of the pyramid $V_{T528}(2^k-1)$.
We combine a piece of the pyramid and the quadrangular prism of size $1 \times 1 \times (2^k-1)$, and call it $B_k$.
Let $C_k$ be the remaining three pieces of the quartered pyramid without the quadrangular prism of size of $1 \times 1 \times (2^k-1)$.

Let $z_k$ be the number of nonzero states in $A_k$. 
We can easily find that the number of nonzero states in $B_k$ is $(z_k-2^k)/4+2^k$ and that the number of nonzero states in $C_k$ is $3(z_k-2^k)/4$.
We set the initial value $z_0 = 1$, and $z_{-1} = 1/2$ because of technical reason. 
We construct $A_{k+1}$ from one $A_k$, two $A_{k-1}$s, eight $B_{k-1}$s, and four $C_k$s (see Figure~\ref{fig:A1}). 
Then, we obtain the following recurrence formula:
\begin{align}
z_{k+1} &= z_k + 2 z_{k-1} + 8 \left( \frac{z_{k-1}-2^{k-1}}{4} + 2^{k-1} \right) + 4 \left( \frac{3(z_k-2^k)}{4} \right)\\
&= 4 (z_k + z_{k-1}).
\end{align}
Set $y_{k+1} = 2 z_k$, and we have
\begin{align}
\left\{
\begin{array}{l l}
z_0 &= 1, \ y_0 = 1, \\
z_{k+1} &= 4 z_k + 2 y_k,\\
y_{k+1} &= 2 z_k.
\end{array}
\right.
\end{align}
For a vector $a$, a matrix $M$, and a vector $v_o$ given by
\begin{align}
a = 
\begin{pmatrix}
1 & 0
\end{pmatrix}, \
M = 
\begin{pmatrix}
4 & 2 \\
2 & 0
\end{pmatrix}, \
u_0 = 
\begin{pmatrix}
z_0 \\
y_0 
\end{pmatrix}
=
\begin{pmatrix}
1 \\
1 
\end{pmatrix}, 
\end{align}
we have
\begin{align}
cum_{T528}(2^k-1) = a M^k u_0 = 2^{k-1} \left((1 + \sqrt{2})^{k+1} + (1 - \sqrt{2})^{k+1} \right).
\end{align}
%
%
Let 
\begin{align}
M_0 = 
\begin{pmatrix}
1 & 0 \\
1 & 0 
\end{pmatrix}, \
M_1 = 
\begin{pmatrix}
3 & 2 \\
1 & 0 
\end{pmatrix}, 
\end{align}
and for the binary number of time step $n = \sum_{i=0}^{k-1} x_{k-i} 2^i$, 
we have
\begin{align}
num_{T528}(n) 
= a M_{x_0} M_{x_1} \cdots M_{x_{k-1}} u_0,
\end{align}
where $x_0=0$.
For $k \geq 0$,
\begin{align}
a M_0^k u_0 &= 1, \\
a M_1^k u_0 &= \frac{17 + 7 \sqrt{17}}{34} \left( \frac{3 + \sqrt{17}}{2} \right)^k + \frac{17 - 7 \sqrt{17}}{34} \left( \frac{3 - \sqrt{17}}{2} \right)^k. 
\end{align}
The matrices $M_0$ and $M_1$ hold the following:
\begin{align*}
a M_1^{k} M_0^{\tilde{k}} u_0 = a M_1^{k} u_0 
&= a M_0^{\tilde{k}} M_1^{k} u_0,\\
a M_1^{k_0} M_0^{\tilde{k_0}} M_1^{k_1} M_0^{\tilde{k_1}} \cdots M_1^{k_l} M_0^{\tilde{k_l}} u_0 
&= (a M_1^{k_0} u_0) (a M_1^{k_1} u_0) \cdots (a M_1^{k_l} u_0).
\end{align*}
Let $p_s$ be the number of clusters of $s$ continuous $1$-states in the binary number of $n=\sum_{i=0}^{k-1} x_{k-i} 2^i$. 
Thus,
\begin{align}
num_{T528}(n) & = \prod_{s=0}^{l+1} (a M_1^s u_0)^{p_s}\\
&= \prod_{s=0}^{l+1} \! \left( \! \frac{17 + 7 \sqrt{17}}{34} \! \left( \! \frac{3 + \sqrt{17}}{2} \right)^s \! \! + \! \frac{17 - 7 \sqrt{17}}{34} \! \left( \! \frac{3 - \sqrt{17}}{2} \! \right)^s \right)^{p_s} \! \!.
\end{align}
\end{proof}

\begin{figure}[htbp]
\begin{minipage}{0.32\linewidth}
\centering
\includegraphics[width=\linewidth]{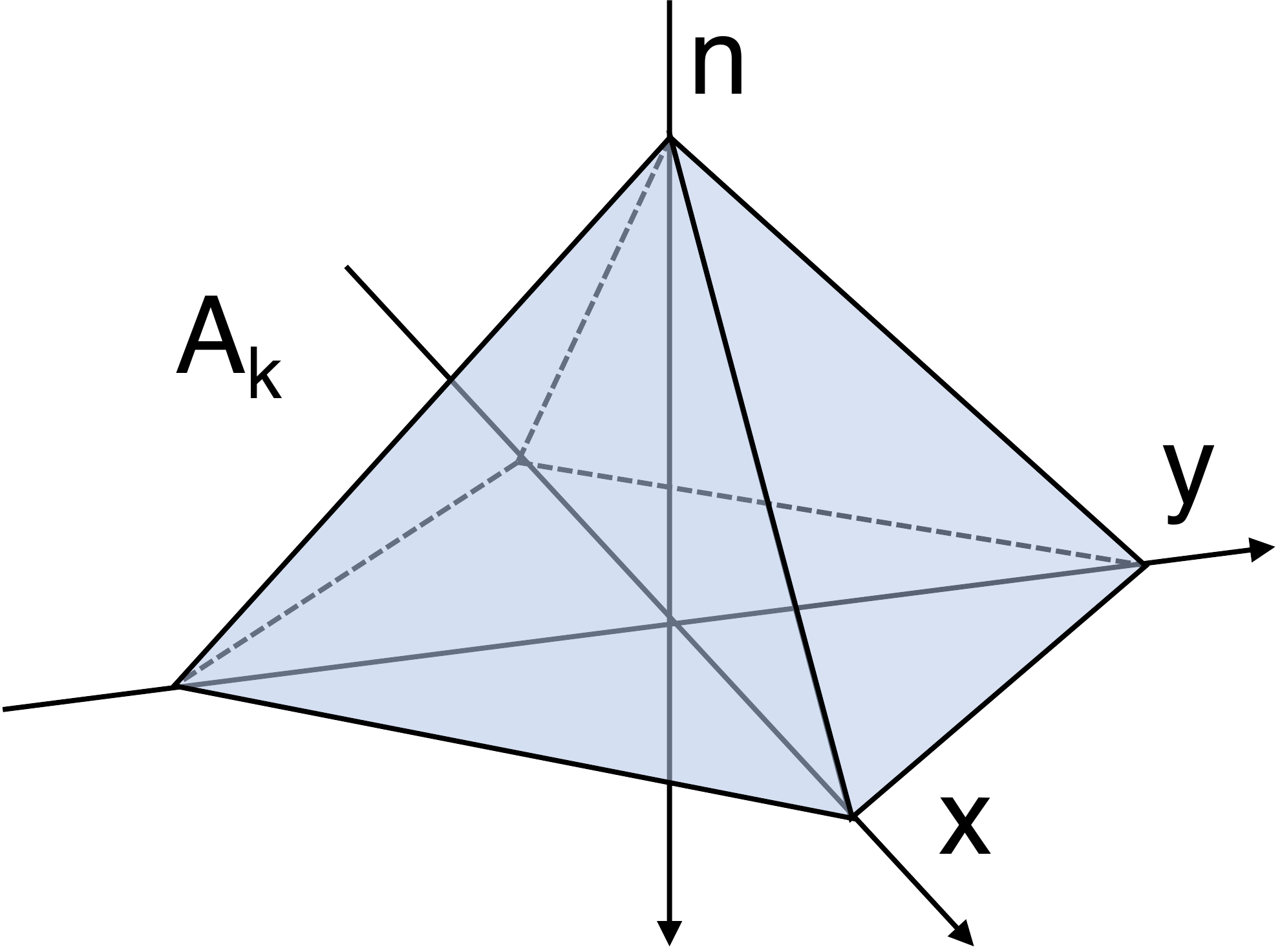}\\
(a) $A_k$
\end{minipage}
\begin{minipage}{0.32\linewidth}
\centering
\includegraphics[width=\linewidth]{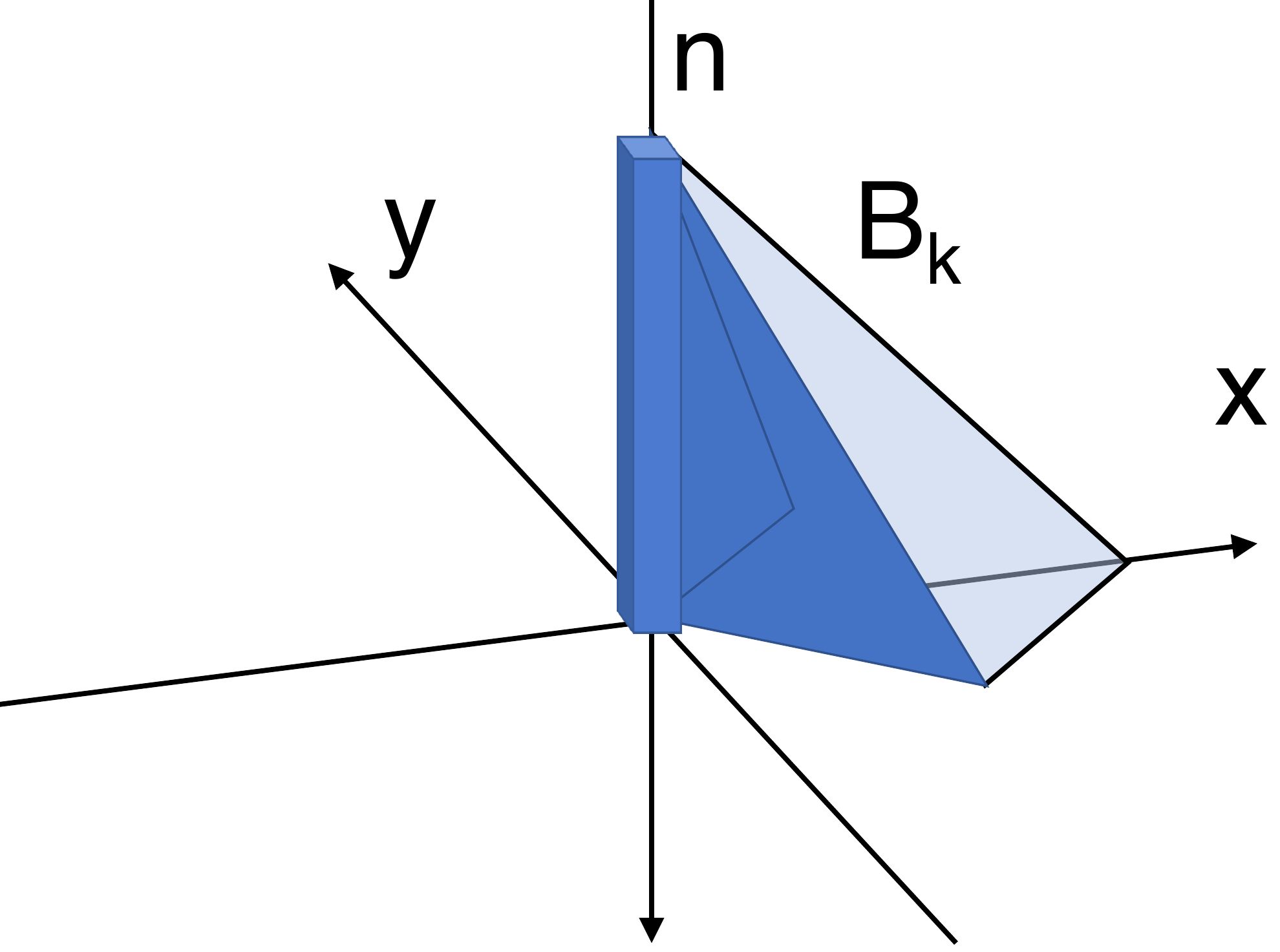}\\
(b) $B_k$
\end{minipage}
\begin{minipage}{0.32\linewidth}
\centering
\includegraphics[width=\linewidth]{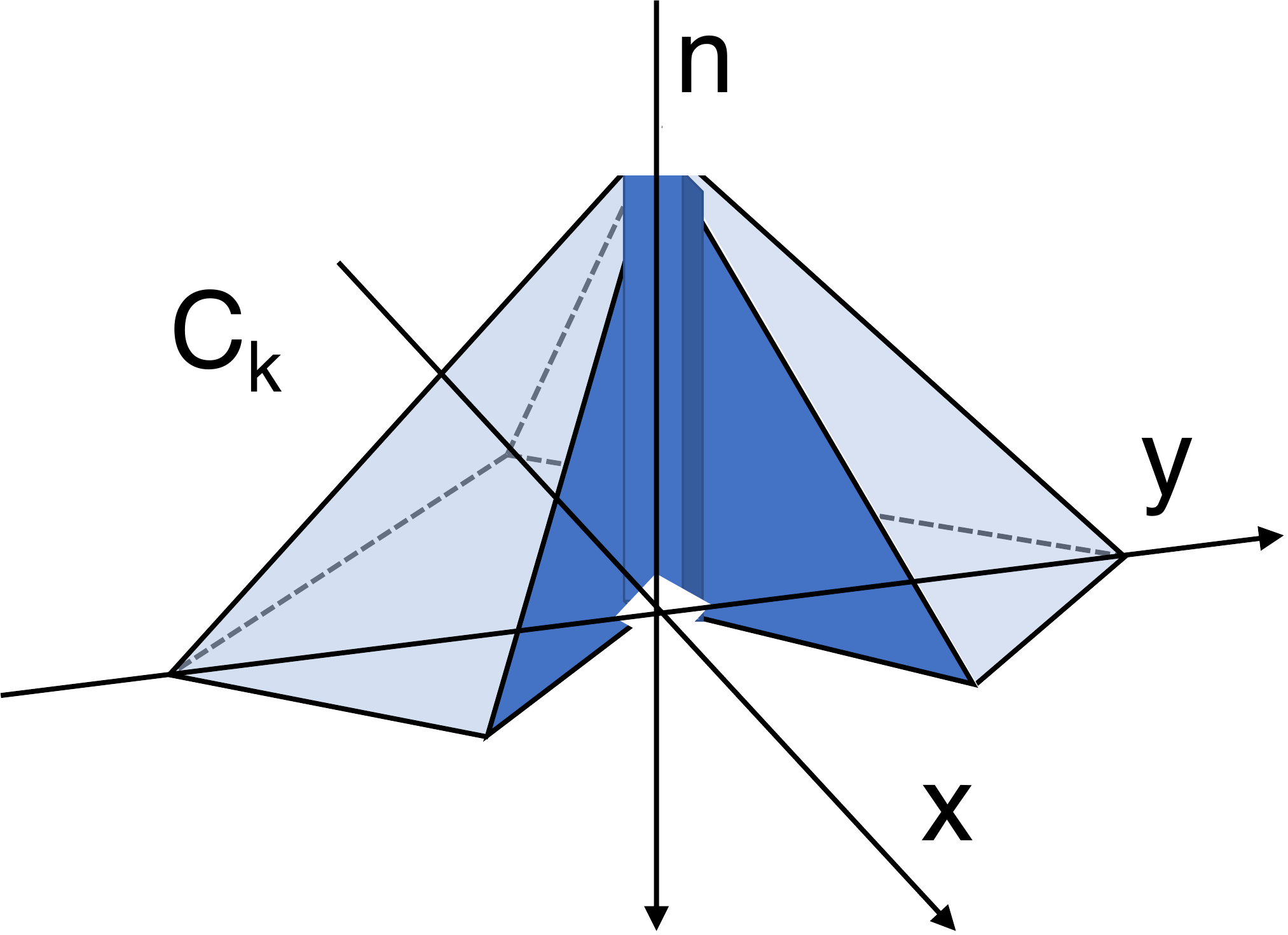}\\
(c) $C_k$
\end{minipage}\\
\quad \\ 
\quad \\
\begin{minipage}{0.32\linewidth}
\centering
\includegraphics[width=\linewidth]{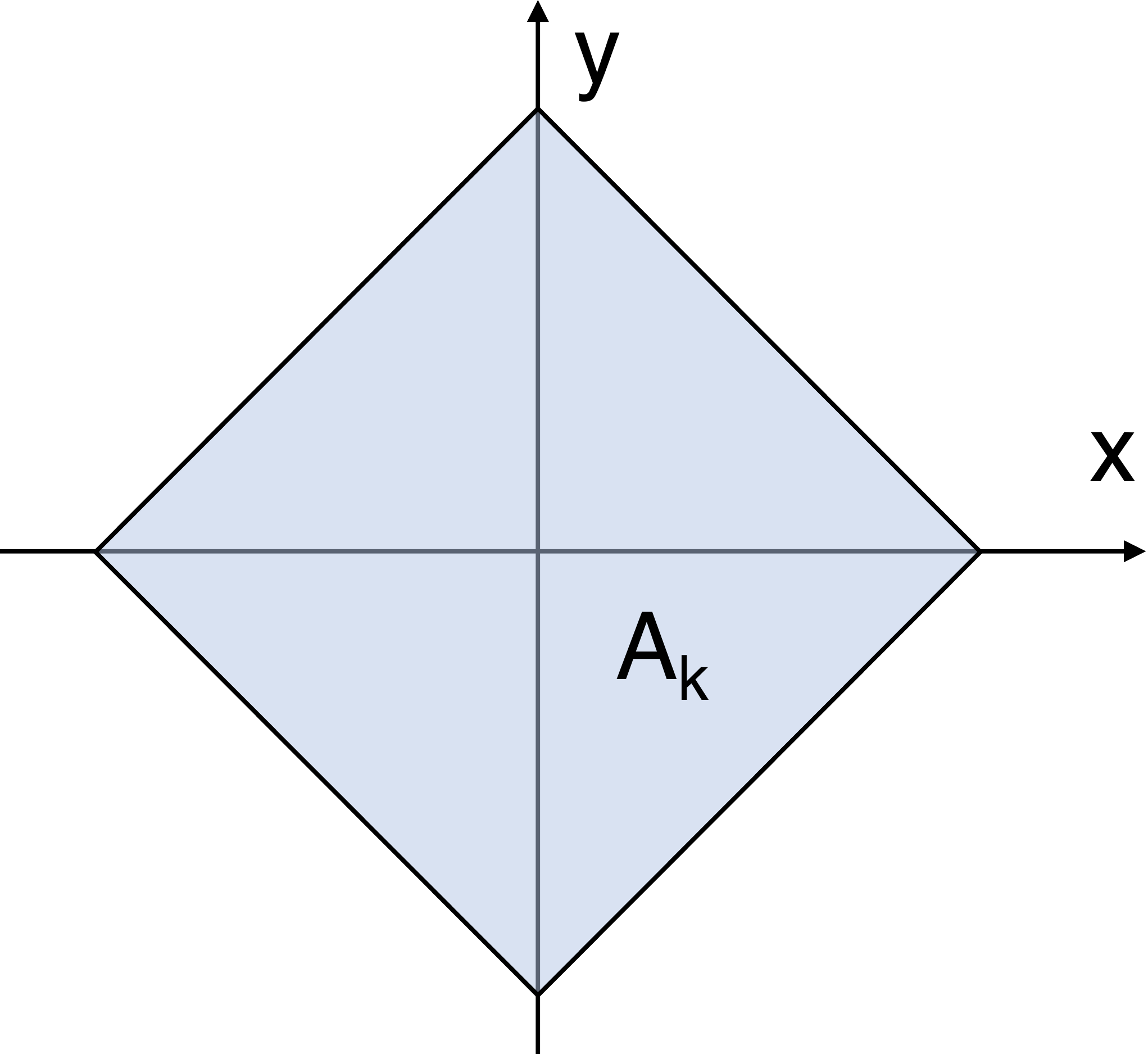}\\
(d) Slice of $A_k$ with $n=2^k-1$
\end{minipage}
\begin{minipage}{0.32\linewidth}
\centering
\includegraphics[width=\linewidth]{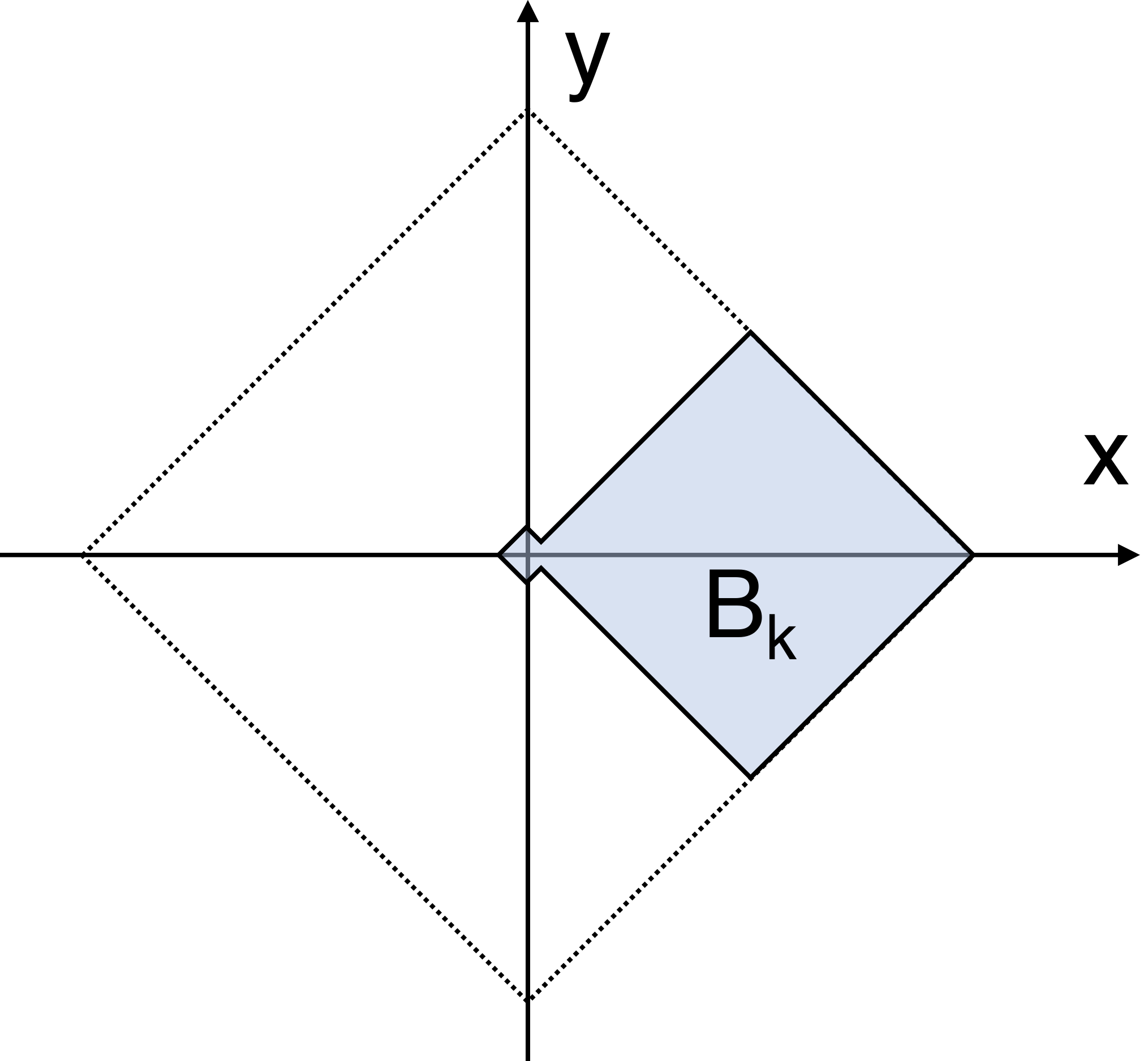}\\
(e) Slice of $B_k$ with $n=2^k-1$
\end{minipage}
\begin{minipage}{0.32\linewidth}
\centering
\includegraphics[width=\linewidth]{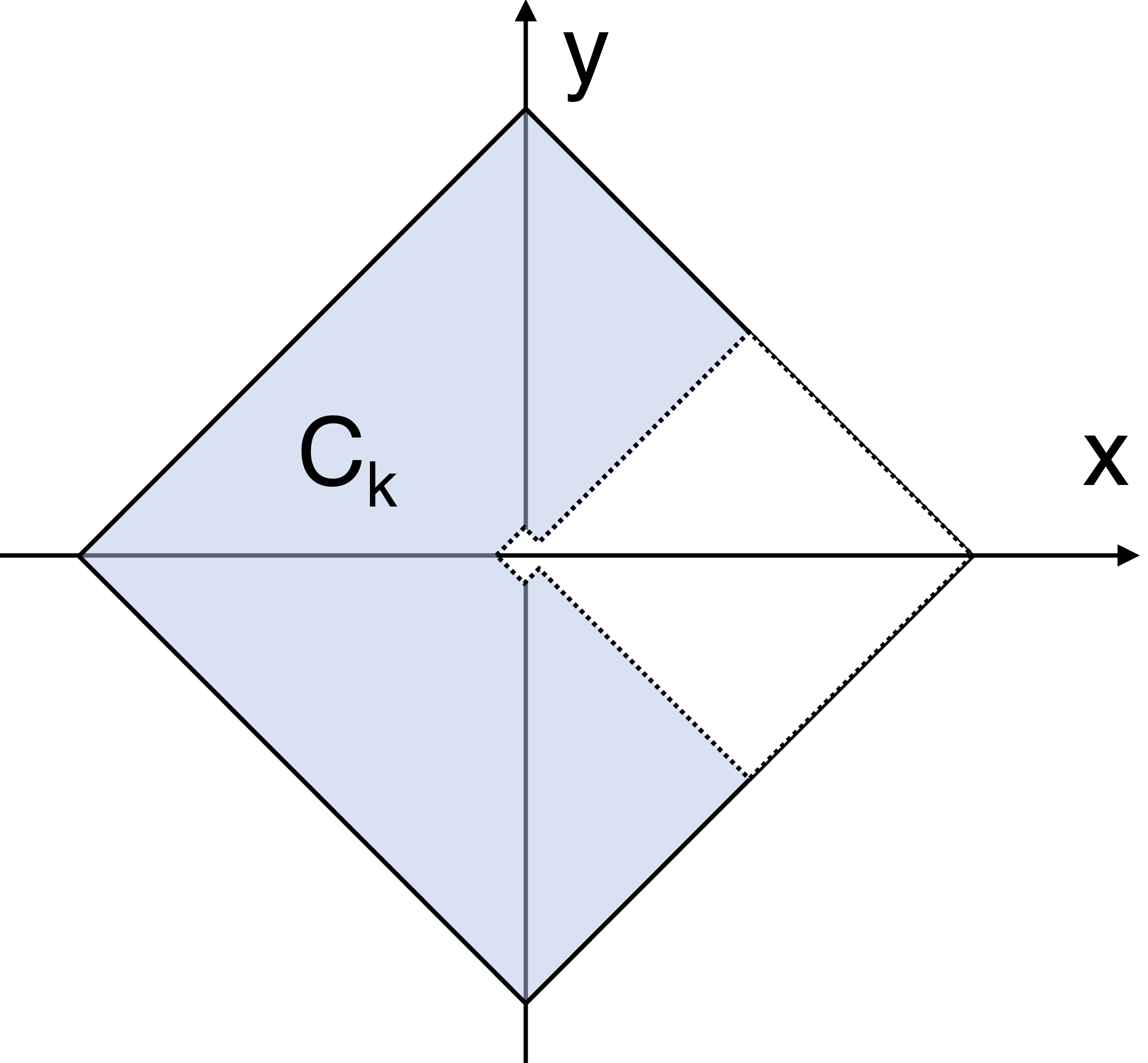}\\
(f) Slice of $C_k$ with $n=2^k-1$
\end{minipage}
\caption{Partially self-similar sets of $\{T_{528}^n u_o\}_{n=0}^{2^k-1}$}
\label{fig:ABC}
\end{figure}

\begin{figure}[htbp]
\begin{minipage}{0.45\linewidth}
\centering
\includegraphics[width=0.9\linewidth]{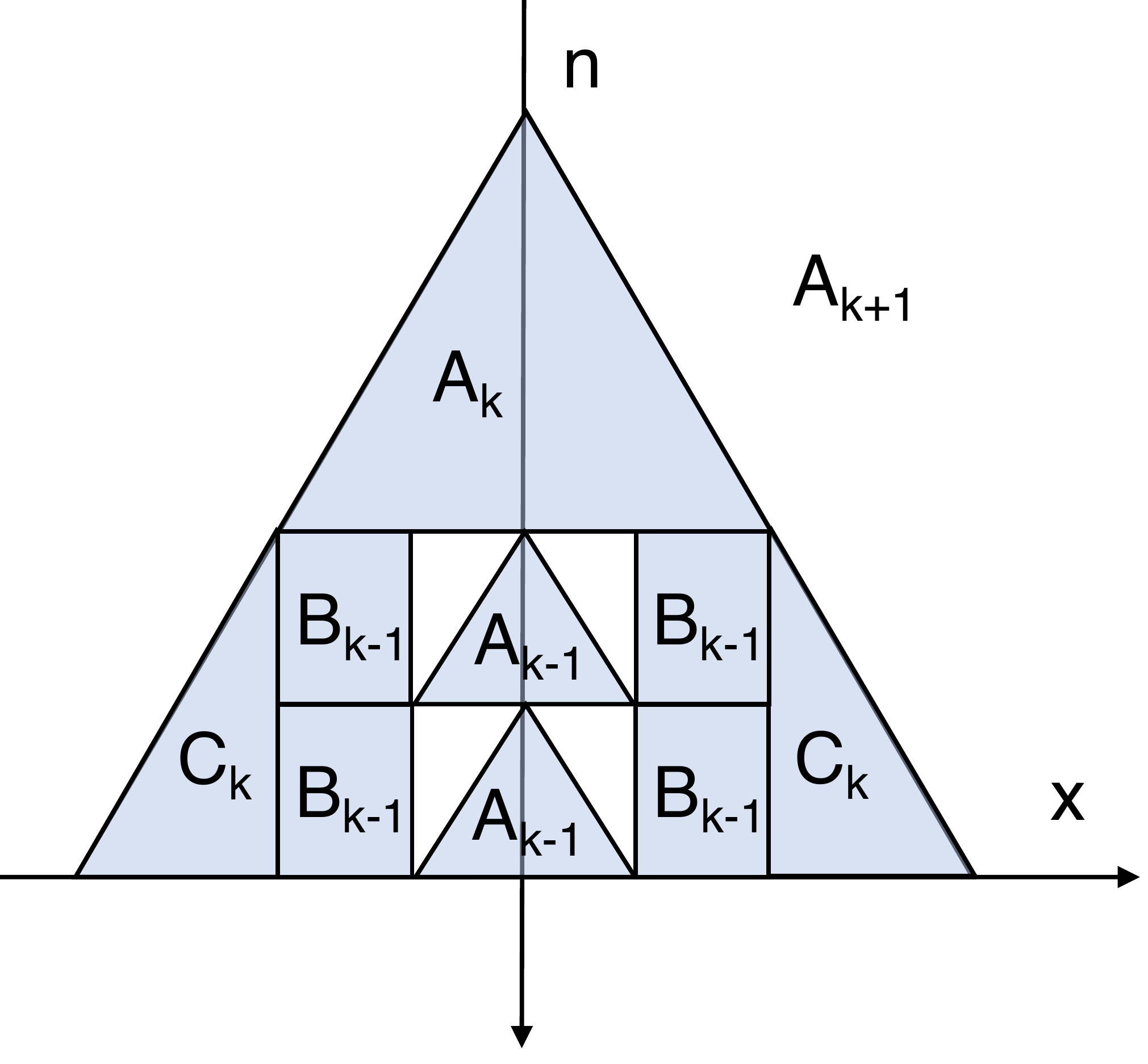}\\
(a) Slice of $A_{k+1}$ with $y=0$
\end{minipage}
\begin{minipage}{0.45\linewidth}
\centering
\includegraphics[width=0.9\linewidth]{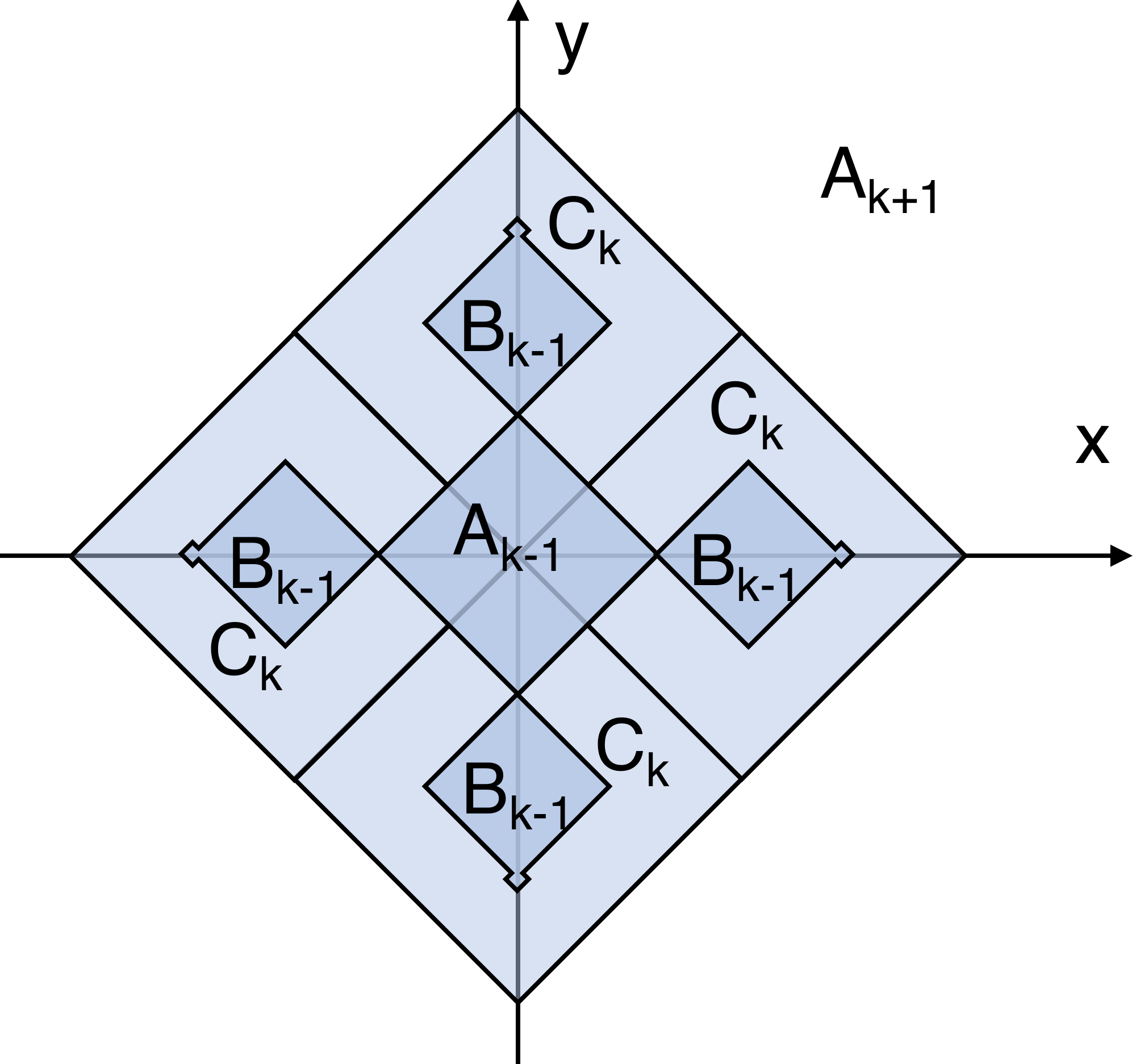}\\
(b) Slice of $A_{k+1}$ with $n=2^k-1$
\end{minipage}
\caption{Self-similar set $A_{k+1}$}
\label{fig:A1}
\end{figure}

Hence, we have the following results for $T_{528}$.

\begin{thm}
\label{thm:fxT528}
For $x = \sum_{i=1}^{\infty} ( x_i / 2^i ) \in [0,1]$, the function $f_{T528} : [0,1] \to [0,1]$ is given by
\begin{align}
\label{eq:defT528}
f_{T528}(x) &\! = \! \sum_{i=1}^{\infty} \!  x_i \alpha^i \! \prod_{s=0}^{i-1} \! \left( \! \frac{17 \! + \! 7 \sqrt{17}}{34} \! \left( \! \frac{3 \! + \! \sqrt{17}}{2} \! \right)^s \! \! + \! \frac{17 \! - \! 7 \sqrt{17}}{34} \! \left( \! \frac{3 \! - \! \sqrt{17}}{2} \! \right)^s \right)^{p_{i,s}} \! \! \! ,
\end{align}
where $\alpha = (\sqrt{2}-1)/2$, and $p_{i, s}$ is the number of clusters consisting of $s$ continuous $1$s in the binary number $0.x_1 x_2 \cdots x_{i-1}$.
\end{thm} 
 
\begin{proof}
By Lemmas~\ref{lem:fin} and \ref{lem:T528num}, for $k \in {\mathbb Z}_{>0}$, we have
\begin{align}
{f_{T528}}_k (x) 
&= \frac{cum_{T528}\left((\sum_{i=1}^k x_i 2^{k-i})-1 \right)}{cum_{T528}(2^k-1)} \\
&= \frac{\sum_{i=1}^k x_i \ num_{T528} \left( \sum_{j=0}^{i-1} x_j 2^{k-j} \right) cum_{T528} (2^{k-i}-1)}{cum_{T528} (2^k-1)} \\
&= \sum_{i=1}^k x_i \, r(x)_i \, \frac{2^{-i} \left((1+\sqrt{2})^{k-i+1}+(1-\sqrt{2})^{k-i+1}\right)}{(1+\sqrt{2})^{k+1} + (1-\sqrt{2})^{k+1}} \\
&= \frac{1}{1+(2\sqrt{2}-3)^{k+1}} \sum_{i=1}^k x_i \, r(x)_i \, {\alpha}^i \nonumber \\
& \quad + \frac{2^{k+1}}{(1+\sqrt{2})^{k+1}+(1-\sqrt{2})^{k+1}} \sum_{i=1}^k x_i \, r(x)_i \frac{(1-\sqrt{2})^{k-i+1}}{2^{k+i+1}}, \label{eq:conv}
\end{align}
where $r(x)_i = a M_{x_0} M_{x_1} \cdots M_{x_{i-1}} u_0$ with $x_0=0$.

Next, in $(a)$ and $(b)$, we show that Equation~\eqref{eq:conv} converges to $\sum_{i=1}^{\infty} x_i \, r(x)_i \, {\alpha}^i$ as $k$ tends to infinity.
In $(a)$, we consider the first term of Equation~\eqref{eq:conv}, and in $(b)$, we consider the second term of Equation~\eqref{eq:conv}.

\begin{enumerate}
\item[$(a)$] We show that the first term of Equation~\eqref{eq:conv} converges $\sum_{i=1}^{\infty} x_i \, r(x)_i \, {\alpha}^i$ as $k \to \infty$.

First, we have $0 \leq x_i \, r(x)_i \alpha^i \leq (4^i/2 + \alpha) \alpha^i$ for any $i > 0$ because
\begin{align}
r(x)_i &\leq a M_0 M_1^{i-1} u_0 \\
&= \frac{17+7 \sqrt{17}}{34} \left( \frac{3+\sqrt{17}}{2} \right)^{i-1} + \frac{17-7 \sqrt{17}}{34} \left( \frac{3-\sqrt{17}}{2} \right)^{i-1} \\
&= \frac{17-\sqrt{17}}{34} \left( \frac{3+\sqrt{17}}{2} \right)^i + \frac{17+\sqrt{17}}{34} \left( \frac{3-\sqrt{17}}{2} \right)^i \\
&= \frac{1}{2} \left( \left( \frac{3+\sqrt{17}}{2} \right)^i + \left( \frac{3-\sqrt{17}}{2} \right)^i \right) \nonumber \\
& \qquad + \frac{\sqrt{17}}{34} \left( \left( \frac{3-\sqrt{17}}{2} \right)^i - \left( \frac{3+\sqrt{17}}{2} \right)^i \right) \\
&< \frac{4^i}{2}+\alpha.
\label{eq:maxrx}
\end{align}
Because $\lim_{i \to \infty} \left| ((4^{i+1}/2 + \alpha) \alpha^{i+1})/((4^i/2 + \alpha) \alpha^i)\right| 
< 1$, the infinite series $\sum_{i=1}^{\infty} (4^i/2 + \alpha) \alpha^i$ absolutely converges.
Thus, $\sum_{i=1}^{\infty} x_i \, r(x)_i \alpha^i$ also absolutely converges.

\item[$(b)$] We show that the second term of Equation~\eqref{eq:conv} converges to $0$ as $k$ tends to infinity.

For the coefficient of the second term of Equation~\eqref{eq:conv}, we easily calculate 
$2^{k+1}/((1+\sqrt{2})^{k+1}+(1-\sqrt{2})^{k+1}) \to 0$ $(k \to \infty)$.
Next, we calculate $\sum_{i=1}^k x_i \, r(x)_i \left((1-\sqrt{2})/2\right)^{k-i+1}$.

When $k$ is even, i.e., $k=2m$ for $m \in {\mathbb Z}_{>0}$, we have
\begin{align}
& \sum_{i=1}^{2m} x_i \, r(x)_i \frac{(1-\sqrt{2})^{2m-i+1}}{2^{2m+i+1}} \\
&= \! \sum_{i=1}^{m} \! x_{2i-1} r(x)_{2i-1} \! \frac{(1\! -\! \sqrt{2})^{2m-(2i-1)+1}}{2^{2m+(2i-1)+1}} \! + \! \sum_{i=1}^{m} \! x_{2i} r(x)_{2i} \frac{(1\! -\! \sqrt{2})^{2m-2i+1}}{2^{2m+2i+1}} \\
&= \! 4 \alpha^{2m+2} \! \sum_{i=1}^m \! x_{2i-1} r(x)_{2i-1} \! \left(\! \frac{\sqrt{2} \! + \! 1}{2} \! \right)^{2i} 
\! \! - \! \alpha^{2m+1} \! \sum_{i=1}^m \! x_{2i} r(x)_{2i} \! \left( \! \frac{\sqrt{2}+1}{2} \!\right)^{2i} \! \!. \label{eq:even}
\end{align}
Here, we evaluate the first term of Equation~\eqref{eq:even}. 
By Equation~\eqref{eq:maxrx},
\begin{align}
&4 \alpha^{2m+2} \sum_{i=1}^{m} x_{2i-1} \, r(x)_{2i-1} \left(\frac{\sqrt{2}+1}{2}\right)^{2i}\\
& \leq 4 \alpha^{2m+2} \sum_{i=1}^{m} \left(\frac{4^{2i-1}}{2}+ \alpha\right)\left(\frac{\sqrt{2}+1}{2}\right)^{2i} \label{eq:even1}\\
&= \frac{\alpha^{2m+2}}{2} \sum_{i=1}^{m} (2\sqrt{2}+2)^{2i} + 4 \alpha^{2m+3} \sum_{i=1}^{m} \left(\frac{\sqrt{2}+1}{2}\right)^{2i}\\
&= \frac{1-\alpha^{2m}}{2(11+8\sqrt{2})} + \alpha \frac{4^{-2m}-\alpha^{2m}}{2 \sqrt{2}-1}\\ 
&\to \frac{8\sqrt{2}-11}{14} 
\quad (m \to \infty).
\end{align}
Because Equation~\eqref{eq:even1} increases with $m$, we have \\
$0 \leq 4 \alpha^{2m+2} \sum_{i=1}^{m} x_{2i-1} \, r(x)_{2i-1} \, ( (\sqrt{2}+1)/2 )^{2i} \leq (8\sqrt{2}-11)/14$.
Next, we evaluate the second term of Equation~\eqref{eq:even}.
By Equation~\eqref{eq:maxrx},
\begin{align}
&\alpha^{2m+1} \sum_{i=1}^{m} x_{2i} \, r(x)_{2i} \, \left( \frac{\sqrt{2}+1}{2}\right)^{2i} \nonumber \\
&\leq \alpha^{2m+1} \sum_{i=1}^{m} \left(\frac{4^{2i}}{2} + \alpha \right) \left( \frac{\sqrt{2}+1}{2}\right)^{2i}\\
\label{eq:even2}
&= \frac{(\sqrt{2}+1)(1-\alpha^{2m})}{11+8\sqrt{2}} +\frac{4^{-2m}-\alpha^{2m}}{4(2 \sqrt{2}-1)}\\
& \to 
\frac{5-3\sqrt{2}}{7} 
\quad (m \to \infty).
\end{align}
Because Equation~\eqref{eq:even2} increases with $m$, we have \\
$0 \leq \alpha^{2m+1} \sum_{i=1}^{m} x_{2i} \, r(x)_{2i} \, ( (\sqrt{2}+1)/2)^{2i} \leq (5-3\sqrt{2})/7$.
Hence, when $k$ is even, 
$-(5-3\sqrt{2})/7 \leq \sum_{i=1}^{2m} x_i \, r(x)_i \, (1-\sqrt{2})^{2m-i+1}/2^{2m+i+1} \leq (8\sqrt{2}-11)/14$.

Next, when $k$ is odd, i.e., $k=2m-1$ for $m \in {\mathbb Z}_{>0}$, we calculate the summation of the second term in Equation~\eqref{eq:conv}. 
We have

\begin{align}
&\sum_{i=1}^{2m-1} x_i \, r(x)_i \frac{(1-\sqrt{2})^{(2m-1)-i+1}}{2^{(2m-1)+i+1}} \\
&= \sum_{i=1}^{m-1} x_{2i} \, r(x)_{2i} \frac{(1-\sqrt{2})^{2m-2i}}{2^{2m+2i}}
+ \sum_{i=1}^m x_{2i-1} \, r(x)_{2i-1} \frac{(1-\sqrt{2})^{2m-(2i-1)}}{2^{2m+(2i-1)}} \\
&= \! \alpha^{2m} \! \sum_{i=1}^{m-1} \! x_{2i} r(x)_{2i} \! \left( \! \frac{\sqrt{2} \! + \! 1}{2} \! \right)^{2i} \! \! - \! 4 \alpha^{2m+1} \! \sum_{i=1}^m \! x_{2i-1} r(x)_{2i-1} \!  \left( \! \frac{\sqrt{2} \! + \! 1}{2} \! \right)^{2i} \! \! .
\label{eq:odd}
\end{align}
We evaluate the first term of Equation~\eqref{eq:odd}. Then, 
\begin{align}
\alpha^{2m} \sum_{i=1}^{m-1} x_{2i} \, r(x)_{2i} \, \left(\frac{\sqrt{2}+1}{2}\right)^{2i} 
&\leq \alpha^{2m} \sum_{i=1}^{m-1} \left(\frac{4^{2i}}{2}+\alpha \right) \left(\frac{\sqrt{2}+1}{2}\right)^{2i}
\label{eq:odd1}\\
&\to 
\frac{8 \sqrt{2}-11}{14} 
\quad (m \to \infty).
\end{align}
Because Equation~\eqref{eq:odd1} is increasing, $0 \leq \alpha^{2m} \sum_{i=1}^{m-1} x_{2i} \, r(x)_{2i} \, ((\sqrt{2}+1)/2)^{2i} \leq (8\sqrt{2}-11)/14$.
We also evaluate the second term of Equation~\eqref{eq:odd}.
Thus, we have
\begin{align}
&4 \alpha^{2m+1} \sum_{i=1}^m x_{2i-1} \, r(x)_{2i-1} \, \left(\frac{\sqrt{2}+1}{2}\right)^{2i}\\
&\leq 4 \alpha^{2m+1} \sum_{i=1}^m \left(\frac{4^{2i-1}}{2} + \alpha \right) \left(\frac{\sqrt{2}+1}{2}\right)^{2i}
\label{eq:odd2}\\
& \to 
\frac{5-3\sqrt{2}}{7} 
\quad (m \to \infty).
\end{align}
Because Equation~\eqref{eq:odd2} is increasing, $0 \leq 4 \alpha^{2m+1} \sum_{i=1}^m x_{2i-1} \, r(x)_{2i-1} \, ((\sqrt{2}+1)/2)^{2i} \leq (5-3\sqrt{2})/7$.
Thus, when $k$ is odd, we have
$-(5-3\sqrt{2})/7 \leq \sum_{i=1}^{2m-1} x_i \, r(x)_i \, 2^{-(2m-1)-i-1} (1-\sqrt{2})^{(2m-1)+1-i} \leq (8 \sqrt{2}-11)/14$.

Therefore, for the second term of Equation~\eqref{eq:conv},
\begin{align}
\frac{2^{k+1}}{(1+\sqrt{2})^{k+1}+(1-\sqrt{2})^{k+1}} \sum_{i=1}^k x_i \, r(x)_i \frac{(1-\sqrt{2})^{k-i+1}}{2^{k+i+1}} \to 0 \quad (k \to \infty),
\end{align}
because $\mid (5-3\sqrt{2})/7 \mid < 1$ and $\mid (8 \sqrt{2}-11)/14 \mid < 1$.\end{enumerate}

By the definition of $f_{T528}$, we verify $f_{T528}(0)=0$ and $f_{T528}(1)=1$.
Then, we have $f_{T528}: [0,1] \to [0,1]$.
\end{proof}

\begin{rmk}
For $x = \sum_{i=1}^k (x_i / 2^i) + 1 / 2^{k+1}$ and
$y = \sum_{i=1}^k ( x_i / 2^i) + \sum_{i=k+2}^{\infty} (1/2^i)$, for $x_i \in \{0,1\}$ and $k \in {\mathbb Z}_{>0}$, we verify that $f_{T528}(x) = f_{T528}(y)$ because $x=y$. Then, 
\begin{align}
&f_{T528}(y)-f_{T528}(x) \nonumber \\
&= \sum_{i=k+2}^{\infty} r(y)_i \alpha^i - r(x)_{k+1} \alpha^{k+1} \\
&= r(x)_{k+1} \alpha^{k+1} \nonumber \\
& \quad \times \left( \sum_{i=1}^{\infty} \! \left( \! \frac{17 + 7 \sqrt{17}}{34} \! \left( \frac{3 + \sqrt{17}}{2} \right)^{i-1} \! \! \! + \! \frac{17 - 7 \sqrt{17}}{34} \!\left( \frac{3 - \sqrt{17}}{2} \right)^{i-1} \right) \! \alpha^i  - 1 \right) \\
&=0.
\end{align}
\end{rmk}

\subsection{$f_{S90}$, $f_{T0}$, and $f_{T528}$ are singular functions}
\label{subsec:rx}

Based on the results in Section~\ref{subsec:selfsim}, we provide a sufficient condition for singularity and show that the resulting functions, $f_{S90}$, $f_{T0}$, and $f_{T528}$, are singular.
We also show that $f_{S90}$ and $f_{T0}$ are Salem's singular function. 

\begin{thm}
\label{thm:sing}
Let $\alpha$ be a parameter such that $0 < \alpha < 1/2$.
For $x = \sum_{i=1}^{\infty} x_i / 2^i \in [0, 1]$, a function $r(x)_i$ is given by $\prod_{s=0}^{i-1} \left( q(s) \right)^{p_{i, s}}$, where $q(s) \in {\mathbb Z}_{>0}$ is a function for $s \in {\mathbb Z}_{\geq 0}$, and $p_{i, s}$ is the number of clusters consisting of $s$ continuous $1$s in the binary number $0.x_1 x_2 \cdots x_{i-1}$.
If $q(s)$ satisfies the following four conditions:
\begin{enumerate}
\item[$(a)$] $q(0)=1$, 
\item[$(b)$] $q(s_1) \, q(s_2-1) < q(s_1 + s_2) \leq q(s_1) \, q(s_2)$ for $s_1$, $s_2 > 0$, 
\item[$(c)$] $q(s+1) \alpha^{s+2} < q(s) \alpha^{s+1} < \sum_{i=s+2}^{\infty} q(i-1) \alpha^i$ for $s > 0$, and
\item[$(d)$] $\sum_{i=1}^{\infty} q(i-1) \, \alpha^i = 1$,
\end{enumerate}
then a function $f(x) = \sum_{i=1}^{\infty} x_i r(x)_i \alpha^i$ for $x = \sum_{i=1}^{\infty} x_i / 2^i \in [0, 1]$ satisfies the following three properties:
\begin{enumerate}
\item[$(i)$] $f$ is strictly increasing,
\item[$(ii)$] $f$ is continuous, and
\item[$(iii)$] $f$ is differentiable with derivative zero almost everywhere.
\end{enumerate}
\end{thm}

\begin{proof}[Proof of Theorem~\ref{thm:sing}~$(i)$]
Suppose $0 \leq x < y \leq 1$. 
We can choose some $k \in {\mathbb Z}_{\geq 0}$ such that 
$x = \sum_{i=1}^k (x_i/2^i) + \sum_{i=k+2}^{\infty} (x_i/2^i)$ and
$y = (\sum_{i=1}^k (x_i/2^i)) + 1/2^{k+1} + (\sum_{i=k+2}^{\infty} (y_i/2^i))$ for $x_i$, $y_i \in \{0, 1\}$, where $\prod_{i=k+2}^{\infty} x_i = 0$. 
When $k=0$, we have $x = \sum_{i=2}^{\infty} (x_i/2^i)$ and
$y = 1/2 + \sum_{i=2}^{\infty} (y_i/2^i)$, where $\prod_{i=2}^{\infty} x_i = 0$. 
Let $\{x\}$ be the fractional part of $x \in {\mathbb R}_{\geq 0}$, i.e., $x-\lfloor x \rfloor$, where $\lfloor x \rfloor$ is the greatest integer less than or equal to $x$.
Thus, we have
\begin{align}
f(y) - f(x) 
&= \left( r(x)_{k+1} \alpha^{k+1} + \sum_{i=k+2}^{\infty} y_i r(y)_i \alpha^i \right) - \left( \sum_{i=k+2}^{\infty} x_i r(x)_i \alpha^i \right)\\
&= r(x)_{k+1} \alpha^{k+1} \left( 1 - \sum_{i=1}^{\infty} x_{i+k+1} r(\{2^{k+1} x\})_i \alpha^i \right) + \sum_{i=k+2}^{\infty} y_i r(y)_i \alpha^i.
\end{align}
Based on the definition $r(x)_{k+1} \alpha^{k+1} > 0$, because $\prod_{i=k+2}^{\infty} x_i=0$ and condition $(d)$, we have $1 - \sum_{i=1}^{\infty} x_{i+k+1} r(\{2^{k+1} x\})_i \alpha^i > 0$ and $\sum_{i=k+2}^{\infty} y_i r(y)_i \alpha^i \geq 0$.
Hence, if $y>x$, then $f(y) > f(x)$.
\end{proof}

\begin{proof}[Proof of Theorem~\ref{thm:sing}~$(ii)$]
Suppose $x=\sum_{i=1}^k (x_i/2^i) + \sum_{i=k+2}^{\infty} (x_i/2^i) \in [0,1)$ and
$y=(\sum_{i=1}^k (x_i/2^i)) + 1/2^{k+1} + (\sum_{i=k+2}^{\infty} (y_i/2^i)) \in (0,1]$,
where $\prod_{i=k+2}^{\infty} x_i = 0$ for some $k \in {\mathbb Z}_{\geq 0}$.
Then $0 < y-x \leq 1/2^k$.
We set $\epsilon := 2 q(1) q(k) \alpha^{k+1}$ and 
$\delta := \epsilon / ( q(k) (2 \alpha)^{k+1})$. 
If $y-x \leq 1/2^k < \delta$, then 
\begin{align}
f(y) - f(x) &= r(x)_{k+1} \alpha^{k+1} \left( 1 - \sum_{i=1}^{\infty} x_{i+k+1} r(\{2^{k+1} x\})_i \alpha^i \right) + \sum_{i=k+2}^{\infty} y_i r(y)_i \alpha^i\\
& \leq r(x)_{k+1} \alpha^{k+1} + \sum_{i=k+2}^{\infty} y_i r(y)_i \alpha^i\\
& \leq r(x)_{k+1} \alpha^{k+1} + r(x)_{k+1} q(1) \alpha^{k+1} \sum_{i=1}^{\infty} y_{i+k+1} r(\{2^{k+1} y\})_i \alpha^i\\
& \leq (1 + q(1)) q(k) \alpha^{k+1} < \epsilon. \label{eq:th4_2}
\end{align}
By $(b)$, we have $q(s_1) q(0) < q(s_1) q(1)$ if $s_2=1$, and then $q(1) > 1$. 
By $(c)$, the sequence $\{ q(k) \alpha^{k+1}\}_{k \geq 1}$ is strictly decreasing.
Thus, Equation~ \eqref{eq:th4_2} is obtained.

Therefore, as $f$ is a function on a finite bounded interval $[0,1]$, $f$ is continuous.
\end{proof}

\begin{proof}[Proof of Theorem~\ref{thm:sing}~$(iii)$]
The function $f$ has bounded variation because $f$ is strictly increasing by Theorem~\ref{thm:sing}~$(i)$.
Hence, $f$ is differentiable almost everywhere on $[0, 1]$ (e.g., \cite[Theorem~6.3.3]{cohn2013}).

Suppose that $x = \sum_{i=1}^{\infty} (x_i/2^i)$ is a differentiable point on $[0, 1]$.
For any $k \in {\mathbb Z}_{>0}$, we can choose $y = \sum_{i=1}^k x_i/2^i$ such that $y \leq x \leq y + 1/2^k$.
Let $l_k$ be $\max \{i \mid x_i=0, 1 \leq i \leq k \}$ if $\prod_{i=1}^k x_i = 0$, and $0$ if $\prod_{i=1}^k x_i = 1$. Then,
\begin{align}
\frac{f ( y + 1/2^k ) - f(y)}{1/2^k}
&= \frac{f \left( y + (\sum_{i=k+1}^{\infty} 1/2^i)\right) - f(y)}{1/2^k}\\
&= 2^k \! \! \sum_{i=k+1}^{\infty} r \! \left( y + \frac{1}{2^k} \right)_i \alpha^i\\
&= 2^k \, r(x)_{l_k} \, \alpha^{l_k} \! \! \sum_{i=k-{l_k}+1}^{\infty} \! \! q(i-1) \alpha^i.
\label{eq:dift1}
\end{align}
Assuming the derivative at $x$ is not zero,
the derivative is finite and positive because $f$ is strictly increasing.
Let $\hat{y} = y + x_{k+1}/2^{k+1}$.
When $x_{k+1}=1$, 
\begin{align}
\frac{2^{k+1} \left( f ( \hat{y} + 1/2^{k+1}) - f(\hat{y}) \right)}{2^k \left( f( y + 1/2^k) - f(y) \right)} 
&= \frac{2 \sum_{i=k-{l_k}+2}^{\infty} q(i-1) \, \alpha^i}{\sum_{i=k-{l_k}+1}^{\infty} q(i-1) \, \alpha^i}\\
&= 2 - \frac{2 \, q(k-{l_k}) \, \alpha^{k-{l_k}+1}}{\sum_{i=k-{l_k}+1}^{\infty} q(i-1) \, \alpha^i}.
\label{eq:dl1}
\end{align}
When $x_{k+1}=0$, 
\begin{align}
\frac{2^{k+1} \left( f ( \hat{y} + 1/2^{k+1}) - f(\hat{y}) \right)}{2^k \left( f( y + 1/2^k) - f(y) \right)} 
&= \frac{2^{k+1} \, r(x)_{l_k} \, q(k-{l_k}) \, \alpha^{k+1}}{2^k \, r(x)_{l_k} \, \alpha^{l_k} \sum_{i=k-{l_k}+1}^{\infty} q(i-1) \, \alpha^i}\\
&= \frac{2 \, q(k-{l_k}) \, \alpha^{k-{l_k}+1}}{\sum_{i=k-{l_k}+1}^{\infty} q(i-1) \, \alpha^i}.
\label{eq:dl2}
\end{align}

By contrast, because $f$ is differentiable at $x$, we have
\begin{align}
\lim_{k \to \infty} \frac{2^{k+1} \left( f ( \hat{y} + 1/2^{k+1}) - f(\hat{y}) \right)}{2^k \left( f( y + 1/2^k) - f(y) \right)} =1.
\end{align}
Based on Equations~(\ref{eq:dl1}) and (\ref{eq:dl2}):
\begin{align}
\lim_{k \to \infty} \frac{2 \, q(k-{l_k}) \, \alpha^{k-{l_k}+1}}{\sum_{i=k-{l_k}+1}^{\infty} q(i-1) \, \alpha^i} = 1,
\end{align}
\begin{align}
\lim_{k \to \infty} \frac{q(k-{l_k}) \, \alpha^{k-{l_k}+1} + \sum_{i=k-{l_k}+2}^{\infty} q(i-1) \, \alpha^i}{2 \, q(k-{l_k}) \, \alpha^{k-{l_k}+1}} = 1,
\end{align}
\begin{align}
\lim_{k \to \infty} \frac{\sum_{i=k-{l_k}+2}^{\infty} q(i-1) \, \alpha^i}{q(k-{l_k}) \, \alpha^{k-{l_k}+1}} = 1.
\end{align}
By $(c)$, for any $s \in {\mathbb Z}_{> 0}$, we have $q(s) \alpha^{s+1} < \sum_{i=s+2}^{\infty} q(i-1) \alpha^i$. 
This contradicts the assumption that the derivative at $x$ is not zero.
Hence, the derivative at $x$ is zero when $f$ is differentiable at $x$.
\end{proof}

\begin{cor}
By Theorem~\ref{thm:sing}, $f_{S90}$, $f_{S150}$, $f_{T0}$, and $f_{T528}$ are singular functions.
\end{cor}

\begin{rmk}
For function $q$ in Theorem~\ref{thm:sing}~$(b)$, the equal sign is used only when $q$ is an exponential function. 
For example, for $f_{S90}$ and $f_{T0}$, the functions are $2^i$ and $4^i$, respectively.
\end{rmk}


\begin{cor}
We can show that the function $f_{S90}$ is Salem's singular function $L_{1/3}$, i.e., 
\begin{align}
f_{S90}(x) &= 
\left\{
\begin{aligned}
&\frac{1}{3} f_{S90}(2x) & \left(0 \leq x < \frac{1}{2} \right),\\
&\frac{2}{3} f_{S90}(2x-1) + \frac{1}{3} & \left(\frac{1}{2} \leq x \leq 1 \right).
\end{aligned}
\right.
\end{align}
We can show that the function $f_{T0}$ is Salem's singular function $L_{1/5}$, i.e., 
\begin{align}
f_{T0}(x) &= 
\left\{
\begin{aligned}
&\frac{1}{5} f_{T0}(2x) & \left(0 \leq x < \frac{1}{2} \right),\\
&\frac{4}{5} f_{T0}(2x-1) + \frac{1}{5} & \left(\frac{1}{2} \leq x \leq 1 \right).
\end{aligned}
\right.
\end{align}
These results match the difference equations in Proposition~\ref{prop:jmp}. 
\end{cor}

\begin{cor}
For a CA $(\{0, 1\}^{{\mathbb Z}^d}, T) \in LS_1 \cup LS_2$, if a function $f_T$ is given by Salem's singular function $L_{1/\alpha}$, the box-counting dimension of the limit set $\lim_{k \to \infty} V_T(2^k-1)/2^k$ is $-\log \alpha/ \log 2$.
\end{cor}

\subsection{Function $f_{S90}$ obtained by nonlinear SPG$1$dECAs}
\label{subsec:NONlin}

In this section, we focus on two nonlinear SPG$1$dECAs, Rule $22$ and Rule $126$.
Sixteen SPG$1$dECAs exist.
Two of them, Rule $90$ and Rule $150$, are linear, and we already discussed them in the previous sections.
The others are nonlinear, and the spatio-temporal patterns of Rule $18$, Rule $146$, and Rule $218$ are the same as that of Rule $90$.
For nonlinear SPG$1$dECAs Rule $50$, Rule $54$, Rule $94$, Rule $122$, Rule $178$, Rule $182$, Rule $222$, Rule $250$, and Rule $254$, the limit sets are not fractals because their box-counting dimensions are $2$.
Thus, this section discusses the other nonlinear SPG$1$dECAs, Rule $22$ and Rule $126$.
Although their spatio-temporal patterns are different from that of Rule $90$ (see Figures~\ref{fig:22p} and \ref{fig:126p}), their resulting functions, $f_{S22}$ and $f_{\hat{S126}}$, are equal to $f_{S90}$.


\begin{figure}[htbp]
\begin{minipage}{0.45\linewidth}
\centering
\includegraphics[width=\linewidth]{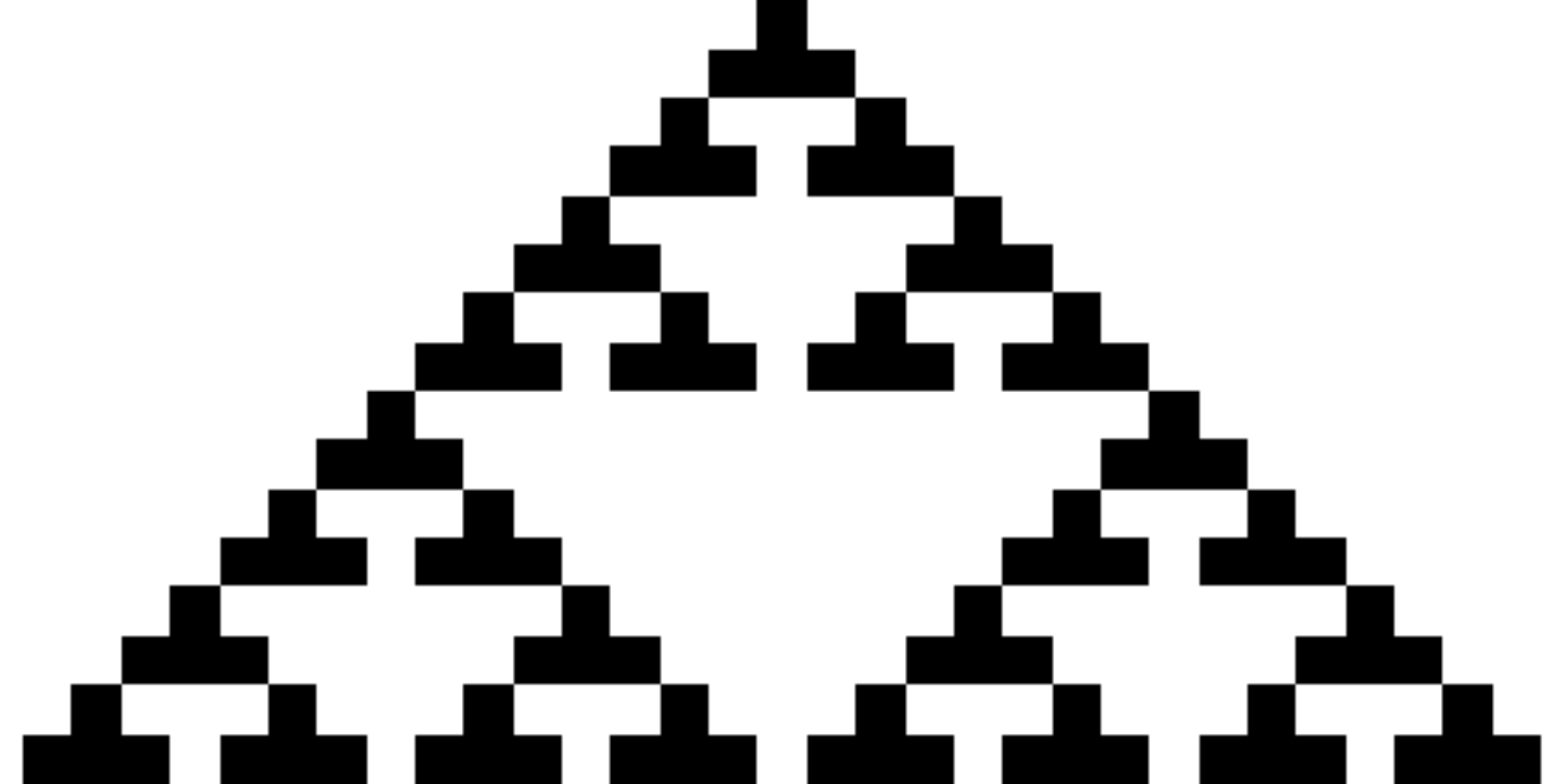}
\caption{Spatio-temporal pattern of Rule $22$}
\label{fig:22p}
\end{minipage}
\begin{minipage}{0.1\linewidth}
\quad
\end{minipage}
\begin{minipage}{0.45\linewidth}
\centering
\includegraphics[width=\linewidth]{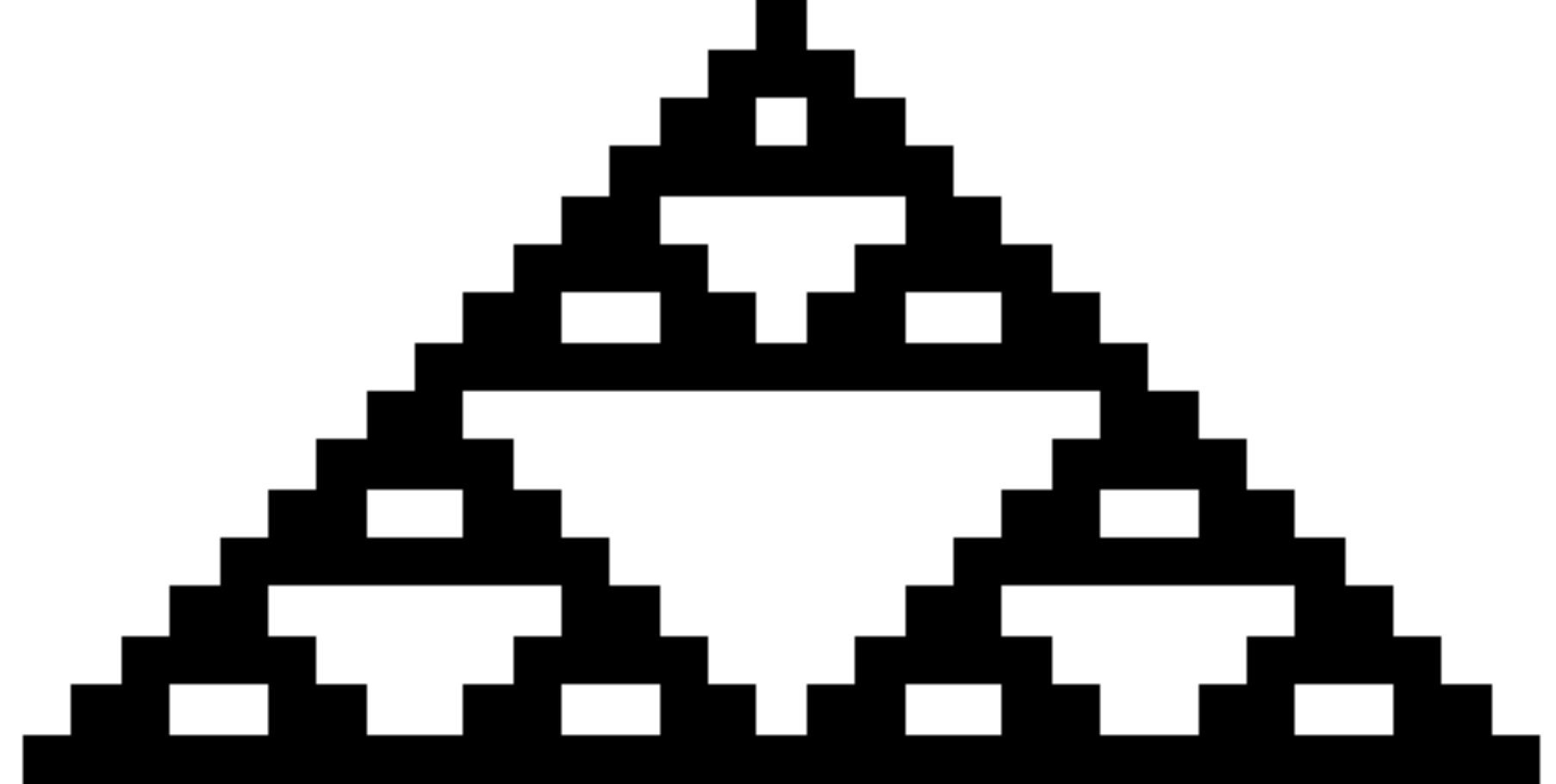}
\caption{Spatio-temporal pattern of Rule $126$}
\label{fig:126p}
\end{minipage}
\end{figure}

\subsubsection{Function $f_{S90}$ by a nonlinear SPG$1$dECA Rule $22$}

Although an SPG$1$dECA Rule $22$ is nonlinear, it holds Lemma~\ref{lem:fin}, and the following results are obtained.

\begin{lem}
\label{lem:22}
Rule $22$ holds the equation in Lemma~\ref{lem:fin}.
\end{lem}

\begin{proof}
The local rule of Rule $22$ is given by $(S_{22} u)_i =  u_{i-1} +u_i + u_{i+1} + u_{i-1} u_i u_{i+1}$ (mod $2$) for $u \in \{0, 1\}^{\mathbb Z}$.
Thus, we have
\begin{align}
(S_{22}^2 u)_i &= (S_{22} (S_{22} u))_i\\
&= u_{i-2} + u_{i-2} u_i + u_{i-2} u_i u_{i+2} + u_i u_{i+2}  + u_{i+2} \nonumber \\
& \quad + u_{i-2} u_{i+1} + u_{i-2} u_{i-1} u_{i+1} + u_{i-2} u_{i-1} u_{i+2} + u_{i-2} u_{i+1} u_{i+2} \nonumber \\
& \quad + u_{i-2} u_{i-1} u_i u_{i+1} u_{i+2} + u_{i-1} u_{i+2} + u_{i-1} u_{i+1} u_{i+2} + u_{i-1} u_i \nonumber \\
& \quad + u_i u_{i+1} 
\quad \mbox{(mod $2$)}. 
\end{align}
If $u_{2m-1}=0$ for any $m \in {\mathbb Z}$, then $(S_{22}^2 u)_{2m-1} = 0$ and 
$(S_{22}^2 u)_{2m} = u_{2m-2} + u_{2m-2} u_{2m} + u_{2m-2} u_{2m} u_{2m+2} + u_{2m} u_{2m+2}  + u_{2m+2}$ (mod $2$).
Inductively, for the odd-numbered columns, we have $(S_{22}^{2k} u)_{2m-1} = 0$ for any $k \in {\mathbb Z}_{\geq 0}$.

Next, if $u_{2m-1}=0$ and $u_{4m}=0$ for any $m \in {\mathbb Z}$, then $(S_{22}^2 u)_{4m-2} = 0$ and $(S_{22}^2 u)_{4m} = u_{4m-2} + u_{4m+2}$ (mod $2$).
Thus, for the even-numbered columns, we have $(S_{22}^2 u)_{2m} = u_{2m-2} + u_{2m+2}$ (mod $2$), which equals the local rule of Rule $90$.
For odd time steps, we easily know $(S_{22}^{2k+1} u)_{2m-1} = (S_{22}^{2k+1} u)_{2m} = (S_{22}^{2k+1} u)_{2m+1}=1$ if and only if $(S_{22}^{2k} u)_{2m}=1$, 
and $(S_{22}^{2k+1} u)_i=0$ otherwise. 

Therefore, the number of nonzero states of Rule $22$ from the initial configuration $u_o$ is given by the equation in Lemma~\ref{lem:fin}.
\end{proof}

\begin{lem}
\label{lem:22num}
For time step $n = \sum_{i=0}^{k-1} x_{k-i} 2^i$, we have $cum_{S22}(2^k-1) = 4 \cdot 3^{k-1}$ and $num_{S22}(n) = 2^{\sum_{i=1}^k x_i} 3^{x_k}$.
\end{lem}

\begin{thm}
\label{thm:fxS22}
For $x = \sum_{i=1}^{\infty} ( x_i / 2^i ) \in [0,1]$, the function $f_{S22} : [0,1] \to [0,1]$ is given by
\begin{align}
\label{eq:defS22}
f_{S22}(x) = f_{S90}(x).
\end{align}
\end{thm} 
 
\begin{proof}
Based on Lemma~\ref{lem:22}, for $x = \sum_{i=1}^{\infty} ( x_i / 2^i ) \in [0,1]$, we have
\begin{align}
f_{S22} (x) 
&= \lim_{k \to \infty} \frac{cum_{S22}\left((\sum_{i=1}^k x_i 2^{k-i})-1 \right)}{cum_{S22}(2^k-1)} \\
&= \lim_{k \to \infty} \frac{\sum_{i=1}^k x_i \ num_{S22} \left( \sum_{j=0}^{i-1} x_j 2^{k-j} \right) cum_{S22} (2^{k-i}-1)}{cum_{S22} (2^k-1)} \\
&= \lim_{k \to \infty} \sum_{i=1}^k x_i \left( 2^{\sum_{j=1}^{i-1} x_j} \right) 3^{-i} \\
&= f_{S90}(x). 
\end{align}
\end{proof}

\subsubsection{Function $f_{S90}$ by a nonlinear SPG$1$dECA Rule $126$}

Rule $126$ is also a nonlinear SPG$1$dECA, and it does not hold Lemma~\ref{lem:fin}.
However, if we consider the spatio-temporal pattern from another initial configuration $\hat{u_o}$, we can consider the function $f_{\hat{S126}}$.

We give a configuration $\hat{u_o} \in \{0, 1\}^{\mathbb Z}$ by $(\hat{u_o})_i=1$ if $i \in \{ 0, 1\}$, and $(\hat{u_o})_i=0$ if $i \in {\mathbb Z} \backslash \{ 0, 1 \}$.
Let $num_{\hat{S126}}(n)$ be the number of nonzero states in a spatial pattern $S_{126}^n \hat{u_o}$ and $cum_{\hat{S126}}(n)$ be the cumulative sum of the number of nonzero states in a spatial pattern $S_{126}^m \hat{u_o}$ from time step $m=0$ to $n$. 
We can now obtain the following relationship similar to Lemma~\ref{lem:fin}.

\begin{lem}
\label{lem:fin126D}
Let $n = \sum_{i=0}^{k-1} x_{k-i} \, 2^i \geq 0$, where $x_0 = 0$. 
For SPG$1$dECA $S_{126}$, we have
\begin{align}
cum_{\hat{S126}} (n-1) 
&= \frac{1}{2} \sum_{i=1}^k x_i \ num_{\hat{S126}} \left( \sum_{j=0}^{i-1} x_j 2^{k-j} \right) cum_{\hat{S126}} (2^{k-i}-1).
\end{align}
\end{lem}

\begin{proof}
The local rule of Rule $126$ is given by $(S_{126} u)_i = u_{i-1} + u_i + u_{i+1} + u_{i-1} u_i + u_i u_{i+1} + u_{i-1} u_{i+1}$ (mod $2$) for $i  \in {\mathbb Z}$.
The differences between the local rule of Rule $126$ and the local rule of Rule $90$ are the transitions for $(u_{i-1}, u_i, u_{i+1}) = (1, 0, 1)$ and $(0, 1, 0)$ (see Table~\ref{tab:1dECA}).

We set the initial configuration $u \in \{0, 1\}^{\mathbb Z}$ such that $u_{2m+1} = u_{2m}$ for any $m \in {\mathbb Z}$.
Then, we have 
\begin{align}
(S_{126} u)_{2m-1} 
&= u_{2m-2} + u_{2m} \ \mbox{(mod $2$)} = (S_{90} u)_{2m-1}, \\
(S_{126} u)_{2m} &= u_{2m-1} + u_{2m+1} \ \mbox{(mod $2$)} = (S_{90} u)_{2m}, \\
(S_{126} u)_{2m+1} &= u_{2m} + u_{2m+2} \ \mbox{(mod $2$)} = (S_{90} u)_{2m+1}, \\
(S_{126} u)_{2m+2} &= u_{2m+1} + u_{2m+3} \ \mbox{(mod $2$)} = (S_{90} u)_{2m+2}.
\end{align}
Because of the assumption of $u$, we have the relationships $(S_{126} u)_{2m-1} = (S_{126} u)_{2m}$ and $(S_{126} u)_{2m+1} = (S_{126} u)_{2m+2}$ for any $m \in {\mathbb Z}$. 
Hence, we show that for the initial configuration $\hat{u_o}$, the spatio-temporal patterns of Rule $126$ and Rule $90$ are the same, and the number of nonzero states of $\{S_{126} \hat{u_o}\}_{n=0}^{2^k-1}$ is double the number of nonzero states of $\{S_{90} u_o\}_{n=0}^{2^k-1}$.
\end{proof}

From the result of Lemma~\ref{lem:fin126D}, we obtain the following result.

\begin{lem}
\label{lem:126Dnum}
$cum_{\hat{S126}}(2^k-1) = 2 \cdot 3^k$ and $num_{\hat{S126}}(n) = 2^{1+\sum_{i=1}^k x_i}$.
\end{lem}

\begin{rmk}
\label{rmk:126num}
If we remove the nonzero states in the center column of the spatio-temporal pattern $\{S_{126}^n \hat{u_o}\}$, we can calculate $cum_{S126}$ and $num_{S126}$.
For time step $n = \sum_{i=0}^{k-1} x_{k-i} 2^i$, we have $cum_{S126}(2^k-1) = 2 \cdot 3^k-k-1$ and $num_{S126}(n) = 4^{1-\prod_{i=1}^k (1-x_i)} \cdot 2^{\sum_{i=1}^{k-1} x_i} - \prod_{i=1}^k x_i$.
\end{rmk}

For Rule $126$, we have the following results.

\begin{thm}
\label{thm:fxS126}
For $x = \sum_{i=1}^{\infty} ( x_i / 2^i ) \in [0,1]$, the function $f_{\hat{S126}} : [0,1] \to [0,2]$ is given by
\begin{align}
\label{eq:defS126}
f_{\hat{S126}}(x) = 2 \, f_{S90}(x).
\end{align}
\end{thm} 
 
\begin{proof}
For $x = \sum_{i=1}^{\infty} ( x_i / 2^i ) \in [0,1]$, we have
\begin{align}
f_{\hat{S126}} (x) 
&= \lim_{k \to \infty} \frac{cum_{\hat{S126}}\left((\sum_{i=1}^k x_i 2^{k-i})-1 \right)}{cum_{\hat{S126}}(2^k-1)} \\
&= \lim_{k \to \infty} \frac{\sum_{i=1}^k x_i \ num_{\hat{S126}} \left( \sum_{j=0}^{i-1} x_j 2^{k-j} \right) \ cum_{\hat{S126}} (2^{k-i}-1)}{cum_{\hat{S126}} (2^k-1)} \\
&= \lim_{k \to \infty} 2 \, \sum_{i=1}^k x_i \left( 2^{\sum_{j=1}^{i-1} x_j} \right) 3^{-i} \\
&= 2 \, f_{S90}(x). 
\end{align}
\end{proof}

\section{Concluding remarks}
\label{sec:cr}

In this paper, we shared our results concerning SPG$1$dECAs and SPG$2$dECAs.
In Section~\ref{subsec:selfsim}, we discussed linear SPG$1$dECAs and linear SPG$2$dECAs.
Because the CAs hold the equation in Lemma~\ref{lem:fin}, we can calculate the numbers of nonzero states of their spatial and spatio-temporal patterns, $num_T$ and $cum_T$, for each CA.
We normalized the numbers and obtained the functions $f_{S90}$, $f_{T0}$, and $f_{T528}$.
In Section~\ref{subsec:rx}, we showed that the functions for linear SPG$1$dECAs and linear SPG$2$dECAS are singular functions that strictly increase, are continuous, and are differentiable with derivative zero almost everywhere.
We provided a sufficient condition of singularity for the function $f$ in Theorem~\ref{thm:sing}. 
From this theorem we showed that $f_{S90}$, $f_{S150}$, $f_{T0}$, and $f_{T528}$ are singular functions.
We also discussed the relationship with Salem's singular function $L_{1/\alpha}$. 
We have $f_{S90} = L_{1/3}$ and $f_{T0} = L_{1/5}$, and the box-counting dimension of their limit sets are $-\log 3/\log 2$ for Rule $90$ and $-\log 5/\log 2$ for $T_0$.
In Section~\ref{subsec:NONlin}, we discussed nonlinear $1$dECAs, specifically Rule $22$ and Rule $126$. From their spatio-temporal patterns, we obtained the functions $f_{S22}$ and $f_{\hat{S126}}$, which equals $f_{S90}$.

In future work, we plan to study the other SPG$2$dECAs.
We will study their number of nonzero states and their normalized functions.
In this paper, we showed that the resulting functions are singular, and in \cite{kawa202103}, we had shown that the functions are discontinuous and Riemann integrable.
We will study other pathological functions, not only singular functions emerging from $2$dECAs, and we will provide generalized conditions for Theorem~\ref{thm:sing}.


\section*{Acknowledgment}
This work was partly supported by a Grant-in-Aid for Scientific Research (18K13457) funded by the Japan Society for the Promotion of Science.

\section*{Data Availability Statement}
The data that supports the findings of this work are available within this paper.


\end{document}